\documentclass[11pt]{article}

\usepackage{natbib}
\usepackage[colorlinks,linkcolor=links,citecolor=cites,urlcolor=MyDarkBlue]{hyperref}
\PassOptionsToPackage{justification=centering}{caption}
\PassOptionsToPackage{}{subcaption}
\PassOptionsToPackage{capitalize, nameinlink}{cleveref}
\usepackage[default]{ezlib}

\usepackage{tikz}
\usepackage{bbm}
\usepackage[section]{placeins}
\usepackage{ragged2e}
\usepackage{marginnote}
\usepackage{abstract}
\usepackage{titling}
\usepackage[titles]{tocloft}

\setlist{noitemsep}

\definecolor{MyDarkBlue}{rgb}{0,0.08,0.45}
\definecolor{cites}{HTML}{324b13}
\definecolor{links}{HTML}{1a663b}
\definecolor{MyLightMagenta}{cmyk}{0.1,0.8,0,0.1}
\definecolor{sblue}{HTML}{0049A9}
\definecolor{scyan}{HTML}{CBEAFC}
\definecolor{sred}{HTML}{B5595C}
\definecolor{sgreen}{HTML}{609B57}
\definecolor{spink}{HTML}{FFB0FF}

\renewcommand{\epsilon}{\varepsilon}
\renewcommand{\phi}{\varphi}

\mathchardef\mathhyphen="2D

\renewcommand{\autoref}{\cref}

\newtoggle{maybemarginnote@reverse}
\togglefalse{maybemarginnote@reverse}
\newrobustcmd{\maybemarginnote}[2][]{%
  \ifbool{vmode}{%
    \@ifnextchar\par{%
      \leavevmode%
      {#2}%
    }{%
      \@ifnextchar\end{%
        \leavevmode%
        {#2}%
      }{%
        {%
          \iftoggle{maybemarginnote@reverse}{%
            \reversemarginpar%
            \global\togglefalse{maybemarginnote@reverse}%
          }{%
            \normalmarginpar%
            \global\toggletrue{maybemarginnote@reverse}%
          }%
          \marginnote{#1{#2}}%
        }%
        \ignorespaces%
      }%
    }%
  }{%
    {%
      \iftoggle{maybemarginnote@reverse}{%
        \reversemarginpar%
        \global\togglefalse{maybemarginnote@reverse}%
      }{%
        \normalmarginpar%
        \global\toggletrue{maybemarginnote@reverse}%
      }%
      \marginnote{#1{#2}}%
    }%
    \ignorespaces%
  }%
}

\makeatother

\definecolor{zorange}{HTML}{D3812F}
\definecolor{zred}{HTML}{C72D27}
\definecolor{zmagenta}{HTML}{A858D1}
\definecolor{zblue}{HTML}{6561C7}
\definecolor{zteal}{HTML}{4494A6}
\definecolor{zgreen}{HTML}{919D24}

\pgfmathsetmacro{\xMax}{5}
\pgfmathsetmacro{\yMax}{3}
\pgfmathsetmacro{\vRv}{1.5}
\pgfmathsetmacro{\vMean}{2}

\pgfmathsetmacro{\vInsp}{\vRv - \vMean * (1 - exp(-(\vRv/\vMean)))}
\pgfmathsetmacro{\vBv}{\vMean * ln(\vMean/\vInsp)}
\newcommand{\fInsp}[1]{(\vInsp + \vMean * (1 - exp(-((#1)/\vMean))))}

\pgfmathsetmacro{\vPHedge}{0.5}

\tikzset{%
  guide/.style={very thin, densely dashed, draw=black!25},
  axis/.style={semithick},
  outside/.style={draw=zmagenta},
  insp/.style={draw=zteal},
  grab/.style={draw=zorange},
  surrogate insp/.style={draw=zred},
  surrogate opt/.style={draw=zred},
  surrogate hedge/.style={draw=zred},
  font=\footnotesize,
  every plot/.style={domain=0:\xMax, samples=128},
}

\makeatletter
\def\@testdef #1#2#3{%
  \def\reserved@a{#3}\expandafter \ifx \csname #1@#2\endcsname
 \reserved@a  \else
\typeout{^^JLabel #2 changed:^^J%
\meaning\reserved@a^^J%
\expandafter\meaning\csname #1@#2\endcsname^^J}%
\@tempswatrue \fi}
\makeatother

\newcommand{\mandator}{obligatory}
\newcommand{\optiona}{nonobligatory}
\newcommand{\Mandator}{Obligatory}

\newcommand{\amandator}{an obligatory}

\newcommand{\mand}{\ensuremath{\mathsf{OI}}}
\newcommand{\opt}{\ensuremath{\mathsf{NI}}}
\newcommand{\hedge}{\ensuremath{\mathsf{LH}}}
\newcommand{\akaobligatoryinspection}{\ifdefstring{\mandator}{obligatory}{\unskip}{(a.k.a. obligatory inspection)}}
\newcommand{\akanonobligatoryinspection}{\ifdefstring{\optiona}{nonobligatory}{\unskip}{(a.k.a. nonobligatory inspection)}}

\newcommand{\dm}{DM}
\newcommand{\iitem}{item}
\newcommand{\Item}{Item}
\newcommand{\items}{items}
\newcommand{\Items}{Items}
\newcommand{\itemindex}{\ensuremath{n}}
\newcommand{\numitems}{\ensuremath{N}}
\newcommand{\itemsset}{\ensuremath{{[N]}}}

\newcommand{\inspcost}{\ensuremath{c}}
\newcommand{\inspcostitem}[1][\itemindex]{\ensuremath{\inspcost_{#1}}}
\newcommand{\realizedcost}{\ensuremath{C}}
\newcommand{\realizedcostsurrogatebase}{\ensuremath{Z}}
\newcommand{\realizedcostsurrogate}{\ensuremath{\realizedcostsurrogatebase^\mand}}
\newcommand{\realizedcostsurrogateopt}{\ensuremath{\realizedcostsurrogatebase^\opt}}
\newcommand{\realizedcostsurrogatehedge}{\ensuremath{\realizedcostsurrogatebase^\hedge}}

\newcommand{\price}{\ensuremath{X}}
\newcommand{\priceitem}[1][\itemindex]{\ensuremath{\price_{#1}}}
\newcommand{\mean}{\ensuremath{\mu}}
\newcommand{\meanitem}[1][\itemindex]{\ensuremath{\mean_{#1}}}
\newcommand{\distribution}{\ensuremath{G}}
\newcommand{\distributionitem}[1][\itemindex]{\ensuremath{\distribution_{#1}}}
\newcommand{\constraints}{\ensuremath{\mathcal{F}}}
\newcommand{\filtration}{\ensuremath{\mathcal{I}}}
\newcommand{\reals}{\ensuremath{\mathbb{R}}}
\newcommand{\terminalcost}{\ensuremath{h}}
\newcommand{\policy}{\ensuremath{\pi}}

\newcommand{\rv}{\ensuremath{u^{\mathsf{rsv}}}}
\newcommand{\bv}{\ensuremath{u^{\mathsf{bkp}}}}

\newcommand{\phedgeitem}[1][\itemindex]{\ensuremath{p_{#1}}}
\newcommand{\phedge}{\ensuremath{p}}
\newcommand{\surrogatebase}{\ensuremath{W}}
\newcommand{\surrogate}{\ensuremath{\surrogatebase^\mand}}
\newcommand{\surrogateopt}{\ensuremath{\surrogatebase^\opt}}
\newcommand{\surrogatehedge}[1][\phedge]{\ensuremath{\surrogatebase^{\hedge(#1)}}}
\newcommand{\surrogateitem}[1][\itemindex]{\ensuremath{\surrogate_{#1}}}
\newcommand{\surrogatebaseitem}[1][\itemindex]{\ensuremath{\surrogatebase_{#1}}}
\newcommand{\surrogateoptitem}[1][\itemindex]{\ensuremath{\surrogateopt_{#1}}}
\newcommand{\surrogatehedgeitem}[1][\itemindex]{\ensuremath{\surrogatehedge[\phedge_{#1}]_{#1}}}
\newcommand{\optsurrogate}{\optiona\ inspection surrogate price}
\newcommand{\mandsurrogate}{\mandator\ inspection surrogate price}

\newcommand{\ratio}{\ensuremath{\alpha}}
\newcommand{\bestratio}{\ensuremath{\alpha^*}}
\newcommand{\approxcombo}{\ensuremath{\beta}}
\newcommand{\policyhedge}{\hedge}
\newcommand{\sur}{surrogate price}
\newcommand{\surs}{surrogate prices}
\newcommand{\inspecteditem}[1][\itemindex]{\ensuremath{\mathbb{I}_{#1}}}
\newcommand{\selecteditem}[1][\itemindex]{\ensuremath{\mathbb{S}_{#1}}}
\newcommand{\selectedset}{\ensuremath{\mathcal{S}}}
\newcommand{\outside}{\ensuremath{r}}
\newcommand{\Policies}{\ensuremath{\Pi}}
\newcommand{\hedgepi}{\ensuremath{{\policyhedge[\pi]}}}
\newcommand{\optpol}{\ensuremath{\Policies^\opt}}
\newcommand{\mandpol}{\ensuremath{\Policies^\mand}}
\newcommand{\policycost}[1][\policy]{\ensuremath{\realizedcost^{#1}}}
\newcommand{\optcost}{\ensuremath{\bar{\realizedcost}^\opt}}
\newcommand{\mandcost}{\ensuremath{\bar{\realizedcost}^\mand}}
\newcommand{\policycostone}[1][\policy]{\ensuremath{\realizedcost^{#1}_{\mathsf{\iitem}}}}
\newcommand{\optcostone}{\ensuremath{\bar{\realizedcost}^\opt_{\mathsf{\iitem}}}}
\newcommand{\mandcostone}{\ensuremath{\bar{\realizedcost}^\mand_{\mathsf{\iitem}}}}
\newcommand{\subproblem}{one-item subproblem}

\newcommand{\subproblems}{one-item subproblems}

\newcommand{\SurMin}[1][m]{R_{\neq #1}}
\newcommand{\SurMinWithout}[1][m]{P_{\neq #1}}
\newcommand{\SurMinWith}[1][m]{Q_{\neq #1}}

\title[Local hedging for Pandora's box with \optiona\ inspection]{Local hedging approximately solves Pandora's box problems with \optiona\ inspection}
\author[Ziv Scully and Laura Doval]{Ziv Scully\thanks{\vspace{-0.225in}Supported by NSF grant no. CMMI-2307008.} \\ Cornell University \and Laura Doval \\ Columbia Business School}
\date{}
\begin{document}

\maketitle

\vspace{0.075in}
\begin{abstract}
We consider search problems with \optiona\ inspection and single-item or combinatorial selection. A decision maker is presented with a number of \items, each of which contains an unknown price, and can pay an inspection cost to observe the \iitem's price before selecting it. Under single-item selection, the decision maker must select one \iitem; under combinatorial selection, the decision maker must select a set of items that satisfies certain constraints. In our \optiona\ inspection setting, the decision maker can select \items\ without first inspecting them. It is well-known that search with \optiona\ inspection is harder than the well-studied \mandator\ inspection case, for which the optimal policy for single-item selection \citep{weitzman1979optimal} and approximation algorithms for combinatorial selection \citep{singla2018price} are known.


We introduce a technique, \emph{local hedging}, for constructing policies with good approximation ratios in the \optiona\ inspection setting. Local hedging transforms policies for the \mandator\ inspection setting into policies for the \optiona\ inspection setting, at the cost of an extra factor in the approximation ratio.
The factor is instance-dependent but is at most~$4/3$. We thus obtain the first approximation algorithms for a variety of combinatorial selection problems, including matroid basis, matching, and facility location.
\end{abstract}

\vspace{-0.075in}
\tableofcontents
\newpage



\section{Introduction}\label{sec:intro}

In the now classic \emph{Pandora's box} problem of \citet{weitzman1979optimal}, a decision maker (henceforth, \dm) possesses $N$ \items, each of which contains an unknown price. \Items\ can be inspected sequentially, at a cost, in any order, and after search, the \dm\ selects one item (or, in generalizations to be discussed shortly, a set of items). The original Pandora's box model has \emph{\mandator\ inspection} \akaobligatoryinspection: when search stops, the \dm\ must select a previously inspected \iitem. In contrast, we consider the problem with \emph{\optiona\ inspection} \akanonobligatoryinspection\ \citep{guha2008information, doval2018whether, beyhaghi2019pandora, beyhaghi2023pandora, fu2023pandora}: when search stops, the \dm\ may select any item, regardless of whether it has been inspected.

The uncertainty in Pandora's box problems is stochastic. Specifically, each \iitem~$\itemindex$ has a known inspection cost~$\inspcostitem$, which is deterministic; and a hidden quantity~$\priceitem$, which is stochastic, drawn from a distribution known to the \dm\ but with unknown realization. Paying the \iitem's inspection cost $\inspcostitem$ reveals the realization of~$\priceitem$. One can consider expected reward maximization or expected cost minimization versions of this problem; we consider the cost-minimization version, and thus call $\priceitem$ the \iitem's \emph{hidden price}. That is, the \dm\ must pay the hidden price of the \iitem(s) they select.

Remarkably, the original \mandator\ inspection Pandora's box model admits an elegant optimal solution. \Items\ are assigned \emph{reservation prices} (see \autoref{eq:reserve} in \autoref{sec:lower-bound}); the \dm\ inspects \items\ in increasing order of their reservation prices; search stops when the minimum inspected price is lower than the minimum reservation price among uninspected \items. Each \iitem's reservation price can be computed ``locally'', with knowledge of that \iitem's hidden price distribution and inspection costs alone. \Citeauthor{weitzman1979optimal}'s policy thus breaks the curse of dimensionality, as the computation required scales only linearly with the number of boxes.

The simplicity of \citeauthor{weitzman1979optimal}'s policy has another benefit: \emph{compositionality}. Specifically, \citet{singla2018price} shows that one can compose \citeauthor{weitzman1979optimal}'s policy with (roughly) any greedy algorithm to solve \emph{combinatorial} Pandora's box problems. These are Pandora's box problems that involve selecting a \emph{set} of \items, not necessarily just one, satisfying some combinatorial constraint. For instance, the Pandora's box version of minimum spanning tree has each \iitem\ sit on an edge of a graph, each containing a hidden price, and the \dm's goal is to select a set of \items\ that forms a spanning tree while minimizing the expected sum of edge prices and inspection costs. \Citet{singla2018price} shows that if a greedy algorithm is a $\approxcombo$-approximation for the deterministic version of the problem, then composing it (in an appropriate sense) with \citeauthor{weitzman1979optimal}'s reservation price policy yields a $\approxcombo$-approximation for the Pandora's box version of the problem with \mandator\ inspection.

Unfortunately, the \optiona\ inspection variant of the Pandora's box problem does not inherit the simple optimal policy of its \mandator\ inspection cousin \citep{doval2018whether}. In particular, finding the optimal policy is known to be NP-hard \citep{fu2023pandora}. While recent work has shown the single-item selection problem admits a PTAS \citep{fu2023pandora, beyhaghi2023pandora}, the algorithm lacks the simplicity and compositionality of \citeauthor{weitzman1979optimal}'s policy for \mandator\ inspection. In particular, there are no approximation algorithms for combinatorial Pandora's box problems with \optiona\ inspection.

In this paper, we define a class of policies, called \emph{local hedging}, for Pandora's box problems with \optiona\ inspection. Local hedging is a \emph{randomized} variant of \citeauthor{weitzman1979optimal}'s reservation price policy. Crucially, local hedging maintains both the simplicity and compositionality of \citeauthor{weitzman1979optimal}'s policy. In the traditional single-item selection problem with cost minimization, local hedging yields a $\frac{4}{3}$-approximation. In the combinatorial setting, we show that for any greedy algorithm for which \citeauthor{singla2018price}'s result \citep{singla2018price} yields a $\approxcombo$-approximation, composing local hedging with that greedy algorithm yields a $\frac{4}{3} \approxcombo$-approximation. Local hedging thus gives the \emph{first approximation algorithms} for combinatorial Pandora's box problems with \optiona\ inspection.

\subsection{Key ideas}

At a high level, the complexity of the \optiona\ inspection Pandora's box problem comes from the fact that each uninspected \iitem\ has two available actions: the \dm\ can inspect it, or the \dm\ can select it without inspection. A natural idea to tame this complexity is to adopt what \citet{beyhaghi2019pandora} call a \emph{committing policy}. Such a policy labels each item as either \mandator-inspection or non-inspection, then obeys these labels. That is:
\* By labeling an \iitem\ as \mandator-inspection, the \dm\ commits to never selecting it unless they inspect it first.
\* By labeling an \iitem\ as non-inspection, the \dm\ commits to never inspecting it.
\*/
The key idea is that any policy obeying these labels is essentially solving \amandator\ inspection problem,\footnote{%
  If \iitem~$\itemindex$ is labeled non-inspection, it is equivalent to \amandator\ inspection \iitem\ with inspection cost~$0$ and hidden price that is deterministically~$\E{\priceitem}$, because the objective is \emph{expected} cost minimization.}
so the results of \citet{weitzman1979optimal} and \citet{singla2018price} can be applied.

It has been shown that in the reward maximization version of single-item selection under \optiona\ inspection, there is a committing policy that achieves a constant-factor approximation \citep{guha2008information, beyhaghi2019pandora}. However, this result is nonconstructive. The resulting algorithm simply evaluates the expected value of all committing policies, then chooses the best one. This is computationally feasible for single-item selection: it is clear that one should never label more than one \iitem\ as non-inspection, so there are only $N + 1$ committing policies to consider. But it is not computationally feasible for combinatorial selection problems. For instance, for $k$-item selection, there are $O(N^k)$ committing policies to consider.

Local hedging is a \emph{randomized} committing policy with an \emph{explicit construction}. Based on each \iitem's inspection cost, mean hidden price, and reservation price, local hedging determines an item-specific hedging probability~$\phedgeitem$. This is ``local'' in the sense that, like \citeauthor{weitzman1979optimal}'s reservation prices, the hedging probabilities do not depend on any other \items' parameters. Under local hedging, the \dm\ simply labels \iitem~$\itemindex$ as \mandator\ inspection with probability~$\phedgeitem$ independently across \items. The \dm\ then applies an appropriate \mandator\ inspection policy, namely the optimal policy from \citet{weitzman1979optimal} for single-item selection, or an approximation algorithm from \citet{singla2018price} for combinatorial selection. We show in \cref{theorem:lh-pandora, theorem:lh-combinatorial} that, roughly speaking, local hedging's randomized commitment inflates the approximation ratio by at most~$\frac{4}{3}$. The $\frac{4}{3}$ is actually a conservative bound on an instance-dependent factor, explained in more detail below.

There are two main technical challenges to defining and analyzing local hedging:
\* Deciding each \iitem's hedging probability.
\* Comparing local hedging's expected cost to that of the intractable optimal policy.
\*/
We solve both problems by defining a new approximation notion, which we call \emph{local approximation} (\cref{definition:local-ratio-approx}). This approximation notion applies to the \emph{\subproblem} in which the \dm\ possesses an \iitem\ and a known outside option, and must select between them. For a fixed hedging probability, local hedging yields a local $\alpha$-approximation if, loosely speaking, randomly committing in the \subproblem\ yields a lower cost than the optimal policy in an inflated \subproblem, in which the inspection cost and hidden price are inflated by a factor of~$\alpha$.


With the notion of local approximation in hand, we provide three main results for single-item selection.
\* We prove a tractable lower bound on the optimal expected cost (\cref{theorem:cost-opt}).\footnote{%
  As we discuss in \cref{sec:literature}, the single-item version of our lower bound is an instance of \emph{Whittle's integral bound} for bandit superprocesses \citep{whittle1980multi, brown2013optimal, aouad2020pandora}. To the best of our knowledge, the extension to the combinatorial setting (\cref{theorem:combinatorial-cost-opt}) is new.}
\* We prove that any \iitem, no matter its inspection cost or hidden price distribution, has a hedging probability that yields a local $\ratio$-approximation for some $\ratio \leq \frac{4}{3}$ (\cref{theorem:lh}).
\* We prove that if all \items\ admit a local $\ratio$-approximation, then local hedging with the associated hedging achieves expected cost at most $\ratio$ times the aforementioned lower bound (\cref{theorem:lh-pandora}).
\*/
The end result is that local hedging is a $(\max_{\itemindex \in [N]} \ratio_\itemindex)$-approximation for single-item selection, where $\ratio_\itemindex$ is the best local approximation ratio achievable for \iitem~$\itemindex$. Our combinatorial results follow the same outline. In fact, only the first and third steps need to change (\cref{theorem:combinatorial-cost-opt, theorem:lh-combinatorial}).

\subsection{Related Literature}\label{sec:literature} Because of its relevance to fundamental applications ranging from innovation to search in online markets,  \cite{weitzman1979optimal} initiated a vast literature studying the problem of sequential inspection  in  Economics, Marketing, Computer Science, and Operations Research (see the surveys by \cite{armstrong2017ordered}, \cite{ursu2023sequential}, and \cite{beyhaghi2023recent}, and \cite{derakhshan2022product} for a recent application to product rankings). Since its inception, several variations to the model have been considered. For instance, \cite{klabjan2014attributes} consider the case of multiple attributes;  \cite{chawla2020pandora} and \cite{gergatsouli2023weitzman} consider the case in which the boxes' contents are correlated; \cite{singla2018price}, \cite{gupta2019markovian}, \cite{boodaghians2020pandora}, \cite{gergatsouli2022online}, and \cite{aminian2023fair} study sequential inspection under various constraints; \cite{hoefer2021stochastic} and \cite{bhaskara2022online} consider the case in which each box can be probed multiple times. Of particular relevance to our work is the aforementioned result of \citet{singla2018price}, who considers selecting multiple items with combinatorial constraints on admissible selection sets. This result was later generalized by \citet{gupta2019markovian} to models where \iitem\ inspection is not an atomic operation, but rather a multi-stage process (see also \citet[Appendix~G]{kleinberg2016descending} and \citet{aouad2020pandora}).

Out of all these variations, the case of nonobligatory inspection has recently captured the attention of researchers in these areas, starting from the work of \cite{guha2008information} and \cite{doval2018whether}. Whereas \cite{doval2018whether} focuses on properties of optimal policies, \cite{guha2008information} provides a 0.8-approximation algorithm. \cite{beyhaghi2019pandora} show that a class of committing policies provides a 0.63-approximation algorithm.  \cite{fu2023pandora} show that computing an optimal policy is NP-Hard. Relying on a new structural property, \cite{fu2023pandora} and \cite{beyhaghi2023pandora} provide polynomial time approximation schemes that for any $\epsilon>0$ compute policies with an expected payoff of at least $(1-\epsilon)$ of the optimal.

Our contribution to the \optiona\ inspection literature is threefold. First, whereas most of the literature focuses on the single-item-rewards case, our local-hedging policy accommodates both costs and combinatorial selection as in \cite{singla2018price}. Second, whereas the aforementioned PTAS results rely on structural properties of the optimal policy for single-item selection under \optiona\ inspection, local hedging only relies on the properties of the much simpler \mandator\ inspection case. It follows that contrary to the existing results, local hedging relies only on calculating two numbers for each box: a reservation value and a hedging probability. Finally, as we discuss in \autoref{sec:conclusions}, local hedging could potentially extend beyond the Pandora's box model with \optiona\ inspection to other so-called \emph{Markovian bandit superprocess} problems \citep{whittle1980multi, glazebrook1982sufficient, gittins2011multi}. Specifically, our notion of local approximation (\cref{definition:local-ratio-approx}), when appropriately generalized, is a relaxation of a condition from the superprocess literature \citep{glazebrook1982sufficient}, and we expect our main lower and upper bound theorems to similarly generalize.

In terms of technical approach, our work is closest to that of \citet{beyhaghi2019pandora}. As previously noted, local hedging is a randomization over the class of committing policies considered by \citet{beyhaghi2019pandora}. Moreover, \citet{beyhaghi2019pandora} also prove a lower bound on the expected optimal cost of single-item selection. Our bound (\cref{theorem:cost-opt}) is always less than theirs, but it is more explicit. See \cref{appendix:beyhaghi} for a detailed comparison. We emphasize, however, that \citet{beyhaghi2019pandora} do not consider the combinatorial case.

Our single-item lower bound is an instance of \emph{Whittle's integral bound}, which is a method of bounding the optimal performance in bandit superprocess problems \citep{whittle1980multi, brown2013optimal, aouad2020pandora}. The closest to ours is an upper bound in \citet[Lemma~1]{aouad2020pandora}, which is an upper bound for a variant of the Pandora's box problem with multiple stages of inspection for each item. Their bound is a translation of a result of \citet[Proposition~4.2]{brown2013optimal} (who build upon \citet[Section~5]{whittle1980multi}) from the discounted bandit superprocess setting to an undiscounted Pandora's-box-type setting. Our single-item lower bound is another translation of the same, with some minor differences (e.g. we treat minimization vs. prior work treating maximization). As such, relative to prior work on Whittle's integral bound, our main contribution is extending it to the combinatorial setting. We also give the bound an elegant probabilistic interpretation of the bound in the ``capped value'' style that is a hallmark of the Pandora's box setting \citep{kleinberg2016descending, beyhaghi2019pandora}.

\subsection{Organization}

The rest of the paper is organized as follows. \autoref{sec:model} introduces the single-item selection model that we use to introduce our results. \autoref{sec:lower-bound} proves a lower bound on the expected cost of the optimal policy (\cref{theorem:cost-opt}). \autoref{sec:local-hedging} introduces local hedging and gives an upper bound on its expected cost (\autoref{theorem:lh-pandora}). \autoref{sec:combinatorial} covers the combinatorial case, introducing the combinatorial-inspection model and extending both our lower and upper bounds (\cref{theorem:combinatorial-cost-opt, theorem:lh-combinatorial}). \autoref{sec:conclusions} concludes with a discussion of how local hedging might extend to settings with reward maximization, as well as settings with more general Markovian bandit superprocesses.

\section{Model}\label{sec:model}

We introduce the single-item selection model, which we then use to state our main results. Though our results hold for more general combinatorial selection problems, the single-item selection model allows us to present the most streamlined version of our results that best emphasizes the key intuition. \Cref{sec:combinatorial} extends our model and results to combinatorial selection.

\paragraph{Single-\iitem\ \optiona\ inspection} A decision maker (\dm) possesses \numitems\ \items, indexed by $\itemindex\in\itemsset=\{1,\dots,\numitems\}$. Each \iitem\ \itemindex\ contains an unknown price, \priceitem, distributed according to \distributionitem, with mean value \meanitem.  The distributions $\{\distributionitem:\itemindex\in\itemsset\}$ are independent.
To observe \iitem\ \itemindex's price, the \dm\ must pay an inspection cost $\inspcostitem\geq0$.

We refer to the tuple $\{(\distributionitem,\inspcostitem):\itemindex\in\itemsset\}$ as an \emph{instance}.
To clarify, the \dm\ knows the instance, but does not know the realization of each $\priceitem \sim \distributionitem$ until after inspecting \iitem~$\itemindex$.


The \dm's task is to adaptively inspect a set of items, paying the inspection costs of each, then ultimately select a single \iitem\ by paying its hidden price. The \dm's goal is to minimize the expected total cost $\E{C}$ of this process, where the total cost~$C$ is given by
\[\label{eq:total-cost}\tag{TC}
  \realizedcost
  = \sum_{\itemindex\in\itemsset} \gp{\selecteditem \priceitem + \inspecteditem \inspcostitem}.
\]
In \cref{eq:total-cost}, \selecteditem\ and \inspecteditem\ are the indicators that \itemindex\ is selected and inspected, respectively.\footnote{We follow the notation in \cite{kleinberg2016descending} and \cite{beyhaghi2019pandora}.} Because the \dm\ must select an item, $\sum_{\itemindex\in\itemsset} \selecteditem = 1$. Note, however, that we do not impose that $\selecteditem \leq \inspecteditem$ and hence, we allow the \dm\ to select an \iitem\ without having first inspected its contents. In other words, our model corresponds to Pandora's box with \optiona\ inspection \citep{guha2008information,doval2018whether}.

Letting \optpol\ denote the set of all (adaptive) policies in the \optiona\ inspection problem, the \dm\ must choose a policy $\policy\in\optpol$ to minimize $\E{\policycost}$, where we let $\policycost$ denote the cost induced by policy~\policy. We additionally denote the optimal expected cost by
\[
  \label{eq:min-optional}\tag{OPT}
  \optcost = \min_{\policy\in\optpol} \E{\policycost}.
\]


\paragraph{\Mandator\ inspection} If, instead, we assume that the \dm\ can only take an \iitem\ it has already inspected (that is, $\selecteditem\leq\inspecteditem$), the above model corresponds to the Pandora's box model in \cite{weitzman1979optimal}. In what follows, we refer to this model as the \mandator\ inspection model and we let \mandpol\ denote the set of policies available to the \dm\ under \mandator\ inspection. Mirroring \cref{eq:min-optional}, we let $\mandcost = \min_{\policy\in\mandpol} \E{\policycost}$.


\begin{remark}[Notational conventions]\label{remark:notation}
We collect in one place our notational conventions. To indicate that we refer to the \optiona\ or \mandator\ inspection models, we label variables with \opt\ or \mand. For instance, $\optpol$ and $\mandpol$ denote the admissible policies under the \optiona\ inspection and \mandator\ inspection models, respectively.
While $\policycost$ and other notations defined later depend on the instance $\{(\distributionitem,\inspcostitem):\itemindex\in\itemsset\}$, we suppress this dependence from on our notation.
\end{remark}

\section{A lower bound on the optimal cost}\label{sec:lower-bound}

In this section, we state and prove \autoref{theorem:cost-opt}, which provides a lower bound on the cost of the optimal policy for the \optiona\ inspection problem. Our lower bound is expressed in terms of each \iitem's \emph{surrogate prices}, defined in \autoref{definition:surrogate}. In order to define surrogate prices, it is helpful to consider a special case of the problem in which the \dm\ is choosing between one \iitem\ and an outside option of known value. We begin with this special case, which we call the \emph{\subproblem}.

Throughout, we focus primarily on \optiona\ inspection, but we discuss \mandator\ inspection when explaining important background, or when it will be important for our later analysis. All of the results for \mandator\ inspection are standard.

\subsection{The \subproblem}
\label{sec:subproblem-definition}
Consider the case in which the \dm\ has a single \iitem\ and an outside option of known value, $\outside\in \mathbb{R}$. One can think of the outside option as a second item that has inspection cost~$0$ and deterministic hidden price~$\outside$. In what follows, to simplify notation, we omit the index \itemindex\ from our notation: the single \iitem\ has hidden price $\price \sim \distribution$ and inspection cost~$\inspcost$.

We denote the expected cost of using a policy~$\policy$ for the \subproblem\ by $\policycostone(\outside)$, and let $\optcostone(\outside) = \min_{\pi \in \optpol} \E{\policycostone(\outside)}$. With that said, for any given \subproblem, there are only three policies that could possibly be optimal:
\*[(a)] Select the outside option. This costs~$\outside$.
\* Inspect the \iitem, then select the better between the \iitem's hidden price~$\price$ and the outside option~$\outside$. This costs $c + \E{\min\{\price, \outside\}}$ in expectation.
\* Select the \iitem\ without inspection. This costs~$\mean$ in expectation.
\*/
The optimal cost for the \subproblem\ is achieved by picking the best among these three, so
\[
  \label{eq:subproblem_cost}
  \optcostone(\outside)
  = \min_{\policy \in \optpol} \E{\policycostone(\outside)}
  = \min\curlgp[\big]{c + \E{\min\{\price, \outside\}}, \outside, \mean}.
\]
It is intuitive that (a) is best for small~$\outside$ and that (c) is best for large~$\outside$, with (b) best at intermediate values. \Citet[Proposition~0]{doval2018whether} formalizes this intuition, characterizing the values of~$\outside$ for which each is optimal (see also \citealp{guha2008information}). We restate the result below, adapting it from rewards to costs.





\begin{lemma}[Proposition 0 in \citealp{doval2018whether}]\label{lemma:one-box}
  Define an \iitem's \emph{reservation price}~$\rv$ and \emph{backup price}~$\bv$ implicitly as follows:
  \begin{align}
    \E_{\distribution}{(\rv - \price)^+} &= \inspcost,\tag{RP}\label{eq:reserve}\\
    \E_{\distribution}{(\price - \bv)^+} &= \inspcost.\tag{BP}\label{eq:backup}
  \end{align}
  For all $\outside\in \mathbb{R}$, the optimal policy in the \subproblem\ is as follows. If $\rv\geq\bv$, then the \dm\ selects the \iitem\ without inspection. Instead, if $\rv<\bv$, the \dm\
  \begin{itemize}
  \item takes the outside option \outside\ if $\outside\leq\rv$,
  \item selects the \iitem\ without inspection if $\outside\geq\bv$, and
  \item otherwise inspects the \iitem\ and selects whatever is best between the \iitem's price \price\ and the outside option \outside.
  \end{itemize}
\end{lemma}

\begin{figure}
  \centering
  \begin{tikzpicture}
  \draw[guide] (\vRv, 0) node[below] {\strut$\rv$} -- (\vRv, \vRv);
  \draw[guide] (0, \vRv) node[left] {$\rv$} -- (\vRv, \vRv);
  \draw[guide] (\vBv, 0) node[below] {\strut$\bv$} -- (\vBv, \vMean);
  \draw[guide] (0, \vMean + \vInsp) node[left] {$\mean + \inspcost$} -- (\xMax, \vMean + \vInsp);

  \draw[axis, ->] (0, 0) -- (\xMax, 0) node[right] {$r$};
  \draw[axis, ->] (0, 0) -- (0, \yMax);

  \draw[ultra thick, densely dotted, outside] (0, 0) -- (\yMax, \yMax);
  \draw[ultra thick, dashed, grab] (0, \vMean) node[left] {$\mean$} -- (\xMax, \vMean);
  \node[left] at (0, \vInsp) {$\inspcost$};
  \draw[ultra thick, insp] plot (\x, {\fInsp{\x}});

  \draw[ultra thick, insp] (\xMax + 1, {\yMax/2 - 0.1 + 0.7}) -- ++(1, 0) node[right] {$\inspcost + \E{\min\{\price, r\}}$\quad(inspect \iitem)};
  \draw[ultra thick, dashed, grab] (\xMax + 1, {\yMax/2 - 0.1}) -- ++(1, 0) node[right] {$\mean$\quad(select \iitem\ without inspection)};
  \draw[ultra thick, densely dotted, outside] (\xMax + 1, {(\yMax/2 - 0.1) - 0.7}) -- ++(1, 0) node[right] {$r$\quad(select outside option)};
\end{tikzpicture}
  \caption{Illustration of the expected cost in the \subproblem\ given each of the three possible first actions. The reservation and backup prices are the values of~$\outside$ that cause indifference between inspecting and the other two actions.}
  \label{fig:one-box}
\end{figure}
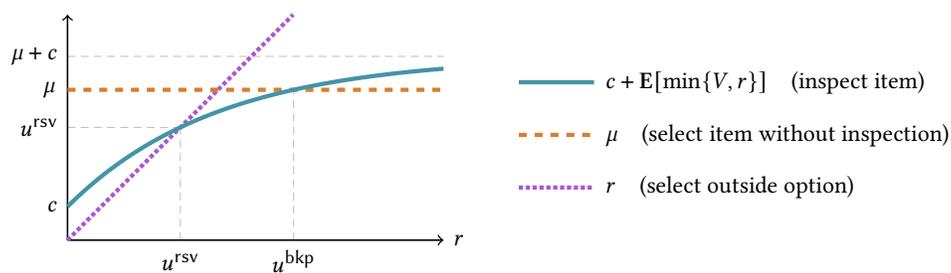

The \iitem's reservation price is the value of the outside option that makes the \dm\ indifferent between taking the outside option and inspecting the \iitem\ (cf. \citealp{weitzman1979optimal}). Similarly, the \iitem's backup price is the value of the outside option that makes the \dm\ indifferent between inspecting the \iitem\ and selecting it without inspection (cf. \citealp{doval2018whether}). See \cref{fig:one-box} for an illustration. The optimal policy for the \subproblem\ follows from these indifference properties.

One can also study the \subproblem\ with \mandator\ inspection. We write $\mandcostone(\outside) = \min_{\policy \in \mandpol} \E{\policycostone(\outside)}$ for the optimal expected cost in this setting. It is a standard result that $\mandcostone(\outside) = \min\{\inspcost + \E{\min\{\price, \outside}, r\}$, analogous to~\cref{eq:subproblem_cost}.

\subsection{Surrogate prices}
Having described the \subproblem, it remains to relate its properties back to the full single-item selection problem with multiple items. We do so by way of \emph{surrogate prices}, defined below. Surrogate prices give a convenient way of characterizing the optimal cost not just in the \subproblem\ (see \cref{lemma:min-r-surrogate}), but also the full single-item selection problem (see \cref{theorem:cost-opt}).


\begin{definition}[Surrogate prices]\label{definition:surrogate}\label{definition:surrogate-optional}
An \iitem's \emph{\mandsurrogate} (henceforth, \mand-\sur), denoted by \surrogate, is the random variable
\begin{align}\label{eq:surrogate-mand}
  \surrogate
  = \max\{\price,\rv\}.\tag{\mand-S}
\end{align}
An \iitem's \emph{\optsurrogate} (henceforth, \opt-\sur), denoted by \surrogateopt, is the random variable
\begin{align}\label{eq:surrogate-opt}
  \surrogateopt
  = \begin{cases}
      \min\{\surrogate, \bv\} & \textnormal{if } \rv < \bv \\
      \mean & \textnormal{if } \rv\geq \bv.
    \end{cases}
\tag{\opt-S}
\end{align}
\end{definition}
The \iitem's \mand-\sur\ captures that in order to obtain the price $\price$, the \dm\ must first inspect the item. By inflating the \iitem's cost up to its reservation price, the \mand-\sur\ \surrogateitem\ internalizes the \iitem's inspection cost. The \opt-\sur\ additionally deflates the \iitem's price down to its backup price \bv. This internalizes that under \optiona\ inspection, the inspection cost may not be paid after all.

Because an \iitem's surrogate price internalizes the inspection cost, we may express the expected optimal cost of the \subproblem\ using the expectation of a \emph{one-shot} choice between the surrogate price and the outside option. This is a standard result in the \mandator\ inspection setting. We extend it to the \optiona\ inspection setting in \cref{lemma:min-r-surrogate} below.

\restatably
\begin{lemma}[Surrogate prices solve one-\iitem\ \optiona\ inspection]\label{lemma:min-r-surrogate}
  For all $\outside \in \reals$,
  \[
    \mandcostone(\outside) = \min\{\inspcost + \E{\min\{\price, \outside\}, \outside}\} &= \E{\min\{\surrogate, \outside\}}, \\
    \optcostone(\outside) = \min\{\inspcost + \E{\min\{\price, \outside\}, \outside, \mean}\} &= \E{\min\{\surrogateopt, \outside\}}.
  \]
\end{lemma}

\begin{proof}[Proof sketch]
  The \mandator\ inspection result is standard, so we focus on \optiona\ inspection. One way to see the result is to observe that $\surrogateopt$ is defined in \cref{eq:surrogate-opt} such that
  \[
    \P{\surrogateopt > \outside} = \frac{\mathrm{d}}{\mathrm{d} \outside} \optcostone(\outside),
  \]
  from which the result follows by integration. See \cref{proof:min-r-surrogate-opt} for the complete proof.
\end{proof}


\subsection{Optimal cost lower bound} Having used the \subproblem\ to introduce \opt-\surs, we are now ready to state our first main result for the full single-item selection problem.


\restatably
\begin{theorem}[Single-item selection lower bound]\label{theorem:cost-opt}
In \optiona\ inspection single-item selection, the optimal policy's expected total cost satisfies
\[\label{eq:lb-opt}\tag{LB-OPT}
  \optcost
  = \min_{\policy\in\optpol} \E{\policycost}
  \geq \E*{\min_{\itemindex\in\itemsset} \surrogateoptitem}.
\]
\end{theorem}

\begin{proof}[Proof sketch]
  Consider an arbitrary policy~$\policy$ for the \dm. It suffices to define a submartingale that is $\E*{\min_{\itemindex\in\itemsset} \surrogateoptitem}$ at time~$0$ and is $\policycost$ when the policy selects an item. The submartingale is
  \[
    K(t) = \realizedcost(t) + \E*{\min_{\itemindex \in \itemsset} \surrogatebaseitem(t) \given \filtration(t)},
  \]
  where
  \[
    \realizedcost(t) &= \textnormal{total inspection and selection cost paid by $\policy$ during $\{0, \dots, t - 1\}$,} \\
    \surrogatebaseitem(t) &=
    \begin{cases}
      0 & \textnormal{if any \iitem\ is selected by $\policy$ during $\{0, \dots, t - 1\}$} \\
      \priceitem & \textnormal{if $\itemindex$ is inspected, but not selected, by $\policy$ during $\{0, \dots, t - 1\}$} \\
      \surrogateoptitem & \textnormal{otherwise,}
    \end{cases} \\
    \filtration(t) &= \textnormal{information $\policy$ gains from inspections during $\{0, \dots, t - 1\}$.}
  \]
  That is, the key idea is to define a time-dependent surrogate price, which is the \opt-\sur\ prior to inspection and the hidden price thereafter. Roughly speaking, $K(t)$ is a submartingale because \opt-surrogate-prices internalize inspection costs, but they do so ``optimistically'', in some sense assuming that an inspected \iitem\ can later be selected without inspection. We formally verify $K(t)$ is a submartingale using \cref{lemma:one-box} and some straightforward computation. See \cref{proof:cost-opt} for the complete proof.
\end{proof}

\begin{remark}
  \Citet[Lemma~16]{beyhaghi2019pandora} prove a result that is similar to \cref{theorem:cost-opt}, giving a bound on $\E{\policycost}$ for any policy $\policy \in \optpol$. Their bound is actually tighter than ours, but it is less explicit, because their bound expression also depends on the policy~$\policy$. We state their bound and discuss it in more detail in \cref{appendix:beyhaghi}, including an alternate proof of \cref{theorem:cost-opt} by way of their bound. This alternate proof is arguably simpler than our main proof above, but we include the main proof regardless because it introduces the key ideas needed for the extension to combinatorial selection.
\end{remark}

\autoref{theorem:cost-opt} states that the expected value of the one-shot problem in which the \dm\ chooses the best \opt-\sur\ is a lower bound for the optimal cost under \optiona\ inspection. In contrast to \autoref{lemma:min-r-surrogate}'s equality, \cref{theorem:cost-opt} is an inequality: the value of the one-shot problem may not be feasible in the \optiona\ inspection problem. This is because, roughly speaking, attempting to achieve the value of the one-shot problem may effectively ask the \dm\ to first inspect a given \iitem, but then later select it without inspection. The value of \cref{theorem:cost-opt} is that the right hand side of \cref{eq:lb-opt} can be computed directly from \subproblems, whereas it is well understood that computing the optimal cost in the \optiona\ inspection problem is intractable.

Anticipating the analysis of the combinatorial case in \autoref{sec:combinatorial}, we note that \autoref{theorem:cost-opt} extends to the combinatorial case (see \autoref{theorem:combinatorial-cost-opt}).

\paragraph{\Mandator\ inspection} It is also useful to contrast \autoref{theorem:cost-opt} with the analogous result for \mandator\ inspection, which involves the \items' \mand-surrogate prices. \Cref{lemma:min-r-surrogate} already suggests that in the \subproblem\ with \mandator\ inspection, the \mand-\sur\ summarizes the value of the optimal policy. In contrast to the case of \optiona\ inspection, this result extends to any number of \items, as stated in \autoref{proposition:no-price-info} below.

\begin{proposition}[Corollary~3 in \citep{beyhaghi2019pandora}]\label{proposition:no-price-info}
In \mandator\ inspection single-item selection, the optimal policy's expected total cost is
\[
  \mandcost
  = \min_{\policy\in\mandpol} \E{\policycost}
  = \E*{\min_{\itemindex\in\itemsset} \surrogateitem}.
\]
\end{proposition}


In other words, despite the adaptive nature of the problem in \cite{weitzman1979optimal}, its value \emph{can} be obtained in a one shot problem, in which the \dm\ picks the \iitem\ with the lowest \mand-\sur. In fact, not only the values of the two problems coincide, but also the \dm\ selects the same \iitem\ under both policies \citep{kleinberg2016descending, beyhaghi2019pandora}. In other words, the one-shot problem also describes the \iitem\ that is \emph{eventually} selected after adaptive inspection in Pandora's box problem.

The contrast between \autoref{theorem:cost-opt} and \autoref{proposition:no-price-info} leaves open the question of whether a \emph{feasible} policy for the \optiona\ inspection case exists that provides a reasonable upper bound for the optimal cost, while at the same time inheriting the simplicity of the policy that delivers the optimal cost of \citet{weitzman1979optimal}.  As we explain next, our local hedging policy does precisely this.

\section{Local hedging}\label{sec:local-hedging}


The local hedging policy that we introduce in this section allows us to marry the results in \autoref{theorem:cost-opt} and \autoref{proposition:no-price-info} to provide a lower bound on the optimal cost for \optiona\ inspection. Local hedging inherits the property of \mandator\ inspection that its value can be calculated from the corresponding surrogate prices. In contrast to the one-shot problem that is defined by the \opt-surrogate prices, the local hedging policy always induces a feasible policy for \optiona\ inspection.

\subsection{Local hedging for the \subproblem}

To define the local hedging policy, we consider first the \subproblem\ (see \cref{sec:subproblem-definition}). In this setting, the local hedging policy is parameterized by a single parameter $\phedge \in [0, 1]$, which we call the item's \emph{hedging probability}. At the beginning, the \dm\ flips a \phedge-weighted coin.
\* With probability~\phedge, the \dm\ labels the \iitem\ as \emph{\mandator-inspection}.
\* With probability~$1-\phedge$, the \dm\ labels the \iitem\ as \emph{non-inspection}.
\*/
These labels constitute commitments by the \dm: a non-inspection item will never be inspected, and \amandator-inspection item will never be selected without inspection. After making this commitment, the \dm\ takes the optimal action that respects this commitment. That is:
\* With probability~\phedge, the \dm\ treats the \subproblem\ as having \mandator\ inspection, choosing between the \iitem\ and the outside option as in \cite{weitzman1979optimal}. Specifically:
\** If $\outside \leq \rv$, the \dm\ selects the outside option~\outside.
\** If $\outside > \rv$, the \dm\ inspects the \iitem, then selects whatever is best between its price~\price\ and the outside option~\outside.
\* With probability~$1-\phedge$, the \dm\ selects whatever is best between the outside option~\outside\ and the \iitem's expected value~$\mean = \E{\price}$.
\*/

Each hedging probability \phedge\ defines a local hedging policy, denoted $\policyhedge(p)$. Below, we sometimes refer to the \emph{local \phedge-hedging policy} when we want to emphasize the specific hedging probability~\phedge.

Like we did for \mandator\ and \optiona\ inspection, we can define an \iitem's local \phedge-hedging \sur.

\begin{definition}[Local hedging surrogate prices]\label{definition:surrogate-hedge}
An \iitem's \emph{local $p$-hedging \sur} (henceforth, \hedge-\sur), denoted by $\surrogatehedge$, is the random variable
\begin{align}\label{eq:surrogate-hedge}\tag{\hedge-S}
  \surrogatehedge
  = \begin{cases}
      \surrogate & \textnormal{with probability } \phedge\\
      \mean & \textnormal{with probability } 1 - \phedge.
    \end{cases}
\end{align}
\end{definition}

See \cref{fig:plot-surrogates} for a comparison between the different types of surrogate prices. The definition of the surrogate price makes evident that under local \phedge-hedging, the \dm\ faces the same problem as in \cite{weitzman1979optimal} with probability \phedge, and with the remaining probability the \dm\ faces an \iitem\ of known value, \mean. Importantly in what follows, the local hedging policy is a policy for the \optiona\ inspection problem that runs by randomizing over policies for the \emph{induced} \mandator\ inspection problem. In particular, \cref{lemma:min-r-surrogate} implies
\[
  \E{\policycostone[\policyhedge(\phedge)](\outside)} = \E{\min\{\surrogatehedge, \outside\}}.
\]

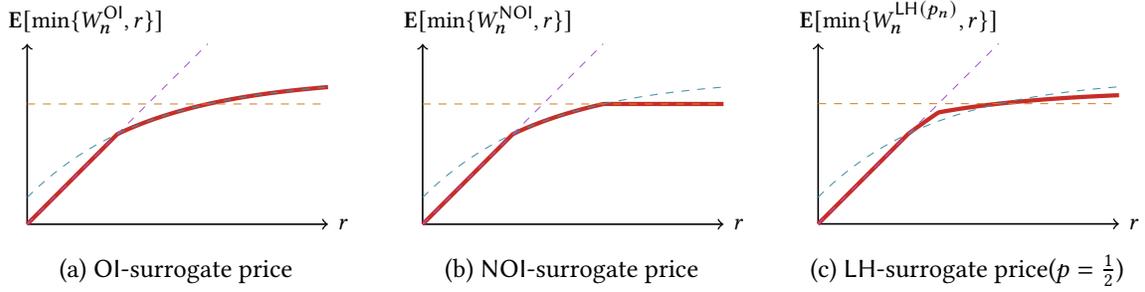
\begin{figure}
  \centering
  \newcommand{\plotSurrogate}[3]{%
  \begin{tikzpicture}[scale=0.8]
    \draw[axis, ->] (0, 0) -- (\xMax, 0) node[right] {$r$};
    \draw[axis, ->] (0, 0) -- (0, \yMax) node[above] {$\mathrlap{\E{\min\{#2, r\}}}\hphantom{\mathbb{EE}}$};

    \draw[ultra thick, surrogate #1] plot (\x, {#3});

    \draw[thin, dashed, outside] (0, 0) -- (\yMax, \yMax);
    \draw[thin, dashed, grab] (0, \vMean) -- (\xMax, \vMean);
    \draw[thin, dashed, insp] plot (\x, {\fInsp{\x}});
  \end{tikzpicture}%
  \ignorespaces%
}%
\begin{subfigure}{0.31\linewidth}
  \centering
  \plotSurrogate{insp}{\surrogate}{min(\x, \fInsp{\x})}
  \caption{\mand-\sur}
\end{subfigure}%
\hfill%
\begin{subfigure}{0.31\linewidth}
  \centering
  \plotSurrogate{opt}{\surrogateopt}{min(\x, \fInsp{\x}, \vMean)}
  \caption{\opt-\sur}
\end{subfigure}%
\hfill%
\begin{subfigure}{0.31\linewidth}
  \centering
  \plotSurrogate{hedge}{\surrogatehedge}{\vPHedge * min(\x, \fInsp{\x}) + (1 - \vPHedge) * min(\x, \vMean)}
  \caption{\hedge-\sur ($\phedge = \frac{1}{2}$)}
\end{subfigure}%
  \caption{Illustrations of the different types of surrogate prices.}
  \label{fig:plot-surrogates}
\end{figure}

\subsection{Local approximation}
A natural question is whether a hedging probability \phedge\ exists such that local hedging provides a good approximation of the optimal policy under \optiona\ inspection. We start by defining our notion of approximation in the single-item case.

\begin{definition}[Local \ratio-approximation]\label{definition:local-ratio-approx}
Consider an \iitem\ with inspection cost~$\inspcost$ and hidden price $\price \sim \distribution$. We say that local \phedge-hedging is a \emph{local \ratio-approximation} for the \iitem\ if for all $\outside\in\reals$,
\[\label{eq:lh:goal}
  \E{\min\{\surrogatehedge, \outside\}} \leq \E{\min\{\ratio\surrogateopt, \outside\}}.
\]
We say the \iitem\ \emph{admits} a local \ratio-approximation if there exists $\phedge \in [0, 1]$ such that local \phedge-hedging is a local \ratio-approximation.
\end{definition}

The intuition behind local \ratio-approximation is as follows. Recall from \cref{lemma:min-r-surrogate} that $\E{\min\{\surrogateopt,\outside\}}$ is the optimal cost in the \subproblem\ under \optiona\ inspection. Thus, \autoref{eq:lh:goal} states that in the \subproblem, given the choice between
\* using local hedging to randomly commit the \iitem\ to be \mandator-inspection or non-inspection; or
\* inflating the \iitem's costs, namely inspection cost and hidden price, by a factor of~$\ratio$;
\*/
local hedging is preferable \emph{for all} outside option values~$\outside$. It is thus conceivable that in the full single-item selection problem, using local hedging for all \items\ is preferable to inflating costs of all \items. Indeed, this is the core idea behind our multi-item result (\cref{theorem:lh-pandora}).


In the multi-item setting the outside option \outside\ is a stand-in for the continuation value after inspecting the current \iitem. The property that the value of the outside option is not scaled up by \ratio\ in \cref{eq:lh:goal} is therefore \emph{crucial}. If we replaced the right-hand side with $\ratio \E{\min\{\surrogateopt, \outside\}}$, we would effectively be asking that using local hedging for a single \iitem\ be preferable to inflating not just that \iitem's costs, but all other \items' costs, too.

Of course, defining local approximation is only useful if it is actually achievable. The main result of this section, \cref{theorem:lh} below, is that all \items\ admit a local $\frac{4}{3}$-approximation or better. We prove \cref{theorem:lh} in \cref{sec:lh-proof}.


\begin{theorem}[All \items\ admit local approximation]\label{theorem:lh}
  Consider an \iitem\ with inspection cost~$\inspcost$ and hidden price $\price \sim \distribution$. Then $\hedge(\phedge)$ is a local $\ratio$-approximation for the \iitem, where
  \[
    \phedge &= \max\curlgp[\bigg]{\frac{\mean - \rv}{\mean - \rv + {\inspcost \rv}/{\mean}}, 0},
    &
    \label{eq:lh:ratio-optimized}
    \ratio &= \max\curlgp[\bigg]{\frac{\mean - \rv + \inspcost}{\mean - \rv + {\inspcost \rv}/{\mean}}, 1} \leq \frac{4}{3}.
  \]
  In particular, all \items\ admit a local $\frac{4}{3}$-approximation.
\end{theorem}

\subsection{Local hedging for single-item selection}

So far, we have shown that local hedging delivers on both our desiderata for the case in which the \dm\ has one item and an outside option with known value. We now show that local hedging's ability to deliver a local \ratio-approximation for the single \iitem\ instances is key for it to deliver a \ratio-approximation to the optimal cost under \optiona\ inspection.

Note first that the local hedging policy can easily be extended to the case in which the \dm\ has $N$ \items. In this case, the policy is parameterized by a vector of hedging probabilities $\mathbf{\phedge}=(\phedge_1,\dots,\phedge_N)$, though we leave this vector implicit in our notation, denoting the policy as simply~$\policyhedge$. At the beginning, the \dm\ independently flips $N$ coins, each with its own bias \phedgeitem. The \dm\ is then faced with an instance of Pandora's box problem, in which the \dm\ can inspect those \items\ that are labeled \mandator-inspection, and conditional on stopping the \dm\ can obtain whatever is best between the already inspected prices and the minimum expected value amongst those \items\ labeled non-inspection.

%

%
%
%

\begin{theorem}[Local hedging solves \optiona\ inspection]\label{theorem:lh-pandora}
  Consider a single-item selection problem with \optiona\ inspection. If every \iitem\ admits a local \ratio-approximation, then by using the corresponding hedging probabilities, local hedging is an \ratio-approximation for single-item selection:
  \begin{align}\label{eq:lh-pandora}
    \E{\policycost[\policyhedge]}
    = \E*{\min_{\itemindex\in\itemsset}\surrogatehedgeitem}
    \leq \ratio \optcost.
  \end{align}
  In particular, there exist hedging probabilities such that local hedging is a $\frac{4}{3}$-approximation.
\end{theorem}
\autoref{theorem:lh-pandora} states that local hedging satisfies our two desiderata in any instance of Pandora's box with \optiona\ inspection. Indeed, the equality in \eqref{eq:lh-pandora} states that the cost of the local hedging policy can be obtained from the \hedge-surrogate prices. Moreover, the inequality in \eqref{eq:lh-pandora} states that local hedging provides a $\frac{4}{3}$-approximation (or better) to the optimal cost under \optiona\ inspection.

\begin{proof}[Proof of \autoref{theorem:lh-pandora}]
  Fix $\ratio$, and let $\phedgeitem$ be the hedging probability that achieves a local $\ratio$-approximation for \iitem~$\itemindex$. The equality $\E{\policycost[\policyhedge]}
    = \E{\min_{\itemindex\in\itemsset}\surrogatehedgeitem}$ follows from \cref{proposition:no-price-info} and the fact that local hedging treats the problem like \amandator-inspection problem after making its randomized commitments. Then, using \cref{eq:lh:goal}, we compute
  \[
    \E*{\policycost[\policyhedge]}
    &= \E*{\min\{\surrogatehedgeitem[1],\surrogatehedgeitem[2],\dots,\surrogatehedgeitem[N]\}} \\
    &\leq \E*{\min\{\ratio\surrogateoptitem[1],\surrogatehedgeitem[2],\dots,\surrogatehedgeitem[N]\}}\\
    &\ \ \vdots \\
    &\leq \E*{\min\{\ratio\surrogateoptitem[1],\ratio\surrogateoptitem[2],\dots,\ratio\surrogateoptitem[N]\}}
    = \ratio\E*{\min_{\itemindex\in\itemsset}\surrogateoptitem}
    \leq \ratio \optcost,
  \]
  where the last inequality follows from \cref{theorem:cost-opt}.
\end{proof}

\paragraph{Organization of the rest of this section}

\autoref{sec:lh-proof} proves \autoref{theorem:lh}, the local approximation step. \autoref{sec:bounds} discusses different ways in which our results on local approximation and local hedging are tight. A reader interested in the application of local hedging to combinatorial optimization can skip straight to \autoref{sec:combinatorial} with little loss of continuity.

\subsection{Proof of \autoref{theorem:lh}}\label{sec:lh-proof}

Before proving \cref{theorem:lh}, we prove a lemma that characterizes the local approximation ratio of local \phedge-hedging for arbitrary~$\phedge$. \Cref{theorem:lh} then follows by optimizing~$\phedge$.

\begin{lemma}\label{lemma:local-ratio-approx}
  Suppose $\rv<\bv$.  For all $\outside \in \reals$, the local \phedge-hedging policy is a local $\ratio(\phedge)$-approximation, where
  \[
    \label{eq:lh:ratio}
    \ratio(\phedge) = 1 + \max\curlgp*{\frac{(1 - \phedge )(\mean - \rv)}{\rv}, \frac{\phedge \inspcost}{\mean}}.
  \]
\end{lemma}

Before proving \autoref{lemma:local-ratio-approx}, let us dissect its statement. The approximation parameter \ratio(\phedge) in \autoref{eq:lh:ratio} admits a natural interpretation. The first term in the maximum describes the \dm's loss in the event the  local \phedge-hedging policy sets the \iitem\ to be a non-inspection \iitem: The \dm\ loses the option to inspect the \iitem\ and hence, the \iitem's reservation price.\footnote{Recall from the definition of the \iitem's \mand-\sur\ that the \iitem's reservation price is the lowest price the \dm\ can hope to obtain taking into account the \iitem's inspection cost.} The second term in the maximum describes the \dm's loss in the event the local \phedge-hedging policy sets the \iitem\ to be \amandator-inspection \iitem: The \dm\ loses the option to take the \iitem\ without inspection, and this loss is larger the larger the \iitem's inspection cost and/or the smaller the \iitem's expected price are. Thus, \autoref{eq:lh:ratio} states that the local \phedge-hedging policy's performance is better the smaller these losses are.

\begin{proof}[Proof of \cref{lemma:local-ratio-approx}]
  Throughout the proof, to simplify notation we omit the dependence of $\ratio(\phedge)$ on \phedge\ and simply denote it by \ratio.

  By \cref{definition:surrogate-hedge, lemma:min-r-surrogate}, showing \cref{eq:lh:goal} amounts to showing the following hold
  \[
    \label{eq:lh:r_small}
    \phedge \, \E{\min\{\surrogate, \outside\}} + (1 - \phedge ) \min\{\mean, \outside\}
    &\leq \E{\min\{\ratio \surrogate, \outside\}},
    \\
    \label{eq:lh:r_large}
    \phedge \, \E{\min\{\surrogate, \outside\}} + (1 - \phedge ) \min\{\mean, \outside\}
    &\leq \ratio \mean.
  \]

  We begin by showing \cref{eq:lh:r_large}. The left-hand side is increasing\footnote{%
    Here and throughout our proofs, we use increasing and decreasing in their weak senses, i.e., to mean ``nondecreasing'' and ``nonincreasing'', respectively.}
  in~\outside\ , but the right-hand side does not depend on~\outside\ . This means it suffices to show \cref{eq:lh:r_large} in the limit when $\outside \to \infty$ . Using that $\E{\surrogate} = \inspcost + \mean$ (see \cref{definition:surrogate}), this reduces to showing
  \[
    \label{eq:lh:constraint_r_large}
    \phedge  \inspcost\leq (\ratio- 1) \mean,
  \]
  which holds for the value $\ratio$ in \cref{eq:lh:ratio}.

  We now show \cref{eq:lh:r_small}. The main obstacle is giving a tight enough bound on $\E{\min\{\ratio \surrogate, \outside\}}$ in terms of $\E{\min\{\surrogate, \outside\}}$. The key observation is that $\E{\min\{\surrogate, \outside\}}$ is concave and increasing as a function of \outside. We also know the graph of the function goes through $(\rv, \rv)$ (see \cref{definition:surrogate}).
  This means that
  \[
    \label{eq:lh:slope}
    m(\outside) = \frac{\E{\min\{\surrogate, \outside\}} - \rv}{\outside - \rv},
  \]
  namely the slope of the line passing through $(\rv, \rv)$ and $(\outside, \E{\min\{\surrogate, \outside\}})$, is decreasing as a function of \outside\ . This implies the lower bound
  \[
    \E{\min\{\ratio \surrogate, \outside\}}
    &= \ratio \gp[\big]{\rv + m\gp[\big]{\tfrac{r}{\ratio}} \, \gp[\big]{\tfrac{r}{\ratio} - \rv}} \\*
    \label{eq:lh:slope_bound}
    &\geq \ratio\gp[\big]{\rv + m(\outside) \, \gp[\big]{\tfrac{r}{\ratio} - \rv}}.
  \]
  Applying \cref{eq:lh:slope, eq:lh:slope_bound} to \cref{eq:lh:r_small}, we find \cref{eq:lh:r_small} holds if
  \[
    m(\outside) \, \gp[\big]{(\ratio - \phedge ) \rv - (1 - \phedge ) r}
    \leq (\ratio - \phedge ) \rv- (1 - \phedge ) \min\{\mean, r\}.
  \]
  Because $m(\outside) \leq 1$ and $-r \leq -\min\{\mean, r\}$, it suffices for the right-hand side to be positive, as then dividing both sides by the right-hand side yields at most~$1$ on the left-hand side. A sufficient condition for the right-hand side to be positive is
  \[
    \label{eq:lh:constraint_r_small}
    (1 - \phedge ) \mean \leq (\ratio - \phedge) \rv,
  \]
  which holds for the value $\ratio$ in \cref{eq:lh:ratio}.
\end{proof}

\begin{proof}[Proof of \autoref{theorem:lh}]
  If $\rv \geq \mean$, then picking $\phedge  = 0$ yields $\ratio = 1$, so the interesting case is when $\rv < \mean$, or equivalently $\rv < \bv$. In this case, the value of $\phedge$ that minimizes $\ratio(\phedge)$ from \cref{lemma:local-ratio-approx} is the value that equalizes the branches of the maximum in \cref{eq:lh:ratio}, which is $\phedge$ from \cref{eq:lh:ratio-optimized}. Computing $\ratio(\phedge)$ then yields the value of $\ratio$ from \cref{eq:lh:ratio-optimized}.

  It remains only to show that the resulting value of $\ratio$ from \cref{eq:lh:ratio-optimized} is at most $\frac{4}{3}$. Let $x = \rv/\mean$. Because $\rv<\mean$, we have $x \in [0, 1)$. This means $\ratio$ is an increasing function of~$\inspcost$. From this and the fact that \cref{eq:reserve} implies $\inspcost < \rv$, we compute
  \[
    \ratio
    = \frac{\mean - \rv + \inspcost}{\mean - \rv + \inspcost x}
    < \frac{\mean - \rv + \rv}{\mean - \rv + \rv x}
    = \frac{1}{1 - x + x^2} \leq \frac{4}{3}.
    \qedhere
  \]
\end{proof}

\subsection{Tightness of local hedging's approximation ratio}\label{sec:bounds}


In this section, we address the question of how tight our bounds on local hedging's approximation ratio are. To frame the discussion more precisely, consider the \subproblem\ with a given \iitem, and define $\bestratio$ to be the minimum value of $\ratio$ such the \iitem\ admits a local \ratio-approximation. We answer the following questions:
\*[(Q1)] For all distributions~$\distribution$ and~$\inspcost$, do we have $\min_{\phedge \in [0, 1]} \ratio(\phedge) = \bestratio$? (Answer: no.)
\* For all $\mean$, $\inspcost$, and~$\rv$, does there exist a distribution~$\distribution$ resulting in the given mean and reservation price such that $\min_{\phedge \in [0, 1]} \ratio(\phedge) = \bestratio$? (Answer: yes.)
\* Do there exist $\inspcost$ and a distribution~$\distribution$ and such that $\bestratio \approx \frac{4}{3}$? (Answer: yes.)
\*/
We expand upon all three answers below. We restrict our attention to the $\rv < \bv$ case, because when $\rv \geq \bv$, we simply have $\ratio(0) = \bestratio = 1$, meaning inspection is never worthwhile. Afterwards, we discuss implications for single-item selection beyond the \subproblem.

The answer to (Q1) is no, but only because of one step in the proof. The value of $\ratio(\phedge)$ in \cref{eq:lh:ratio} is the minimum value that satisfies constraints \cref{eq:lh:constraint_r_large, eq:lh:constraint_r_small}. So (Q1) reduces to: must $\bestratio$ satisfy \cref{eq:lh:constraint_r_large, eq:lh:constraint_r_small}? One can show that \cref{eq:lh:constraint_r_large} is necessary by looking at the $r \to \infty$ limit, but \cref{eq:lh:constraint_r_small} is not. This is because the computation that leads to \cref{eq:lh:constraint_r_small} involves \cref{eq:lh:slope_bound}, an inequality which need not be tight.

The above discussion suggests an answer to (Q2): we have $\bestratio = \min_{\phedge \in [0, 1]} \ratio(\phedge)$ if \cref{eq:lh:slope_bound} holds with equality. In fact, a closer inspection of the proof reveals that we only need \cref{eq:lh:slope_bound} to be tight when $r = \mean$. A straightforward computation shows that this holds when
\[
  \label{eq:bad-bimodal}
  \price =
  \begin{cases}
    0 & \textnormal{with probability } \frac{\inspcost}{\rv} \\
    \frac{\rv \mean}{\rv - \inspcost} & \textnormal{with probability } \frac{\rv - \inspcost}{\rv}.
  \end{cases}
\]

The example in \cref{eq:bad-bimodal} also addresses (Q3). If we choose $\mean = 2 \rv = (2 - \epsilon) \inspcost$ in \cref{eq:bad-bimodal}, then by the above discussion and \cref{eq:lh:ratio-optimized}, we have $\bestratio = \max_{\phedge \in [0, 1]} \ratio(\phedge) = \frac{4 + \epsilon}{3 + \epsilon}$, where $\epsilon$ may be arbitrarily small. The fact that the worst-case scenario is when $\price$ is a high-variance two-point distribution is unsurprising in light of the results of \citet{beyhaghi2019pandora}, who use a similar construction to bound the approximation ratio of committing policies in the reward maximization setting.

Finally, let us zoom out from the \subproblem\ to full single-item selection. In some sense, our answer to (Q3) yields single-item selection problems for which local hedging is arbitrarily close to a $\frac{4}{3}$-approximation, because the \subproblem\ is a special case of single-item selection. But this is somewhat unsatisfying given the optimal policy for the \subproblem\ is known. We can obtain a more satisfying example by using two \items. \Item~1 is as in \cref{eq:bad-bimodal} and the above answer to (Q3), and \iitem~2 has inspection cost~$\inspcost_2 = \epsilon$ and has hidden price equally likely to be $\price_2 = \mu_1$ or $\price_2 = \mu_1/\epsilon$. For small enough~$\epsilon$, it is clearly optimal to first inspect \iitem~2, after which the problem becomes a \subproblem\ with \iitem~1 and outside option~$\price_2$. Following \cref{lemma:one-box}, the optimal policy then either inspects the first \iitem\ (if $\price_2 = \mu_1$) or selects the first \iitem\ without inspection (if $\price_2 = \mu_1/\epsilon$). Local hedging, and indeed any committing policy, does worse because it must label \iitem~1 as \mandator-inspection or non-inspection before learning~$\price_2$. This construction closely mirrors that of \citet[Example~1]{beyhaghi2019pandora}.



\section{Combinatorial Pandora's box problems}\label{sec:combinatorial}

We show in this section how local hedging can also be used to provide approximately optimal policies in combinatorial versions of the \optiona\ inspection problem. To this end, we consider the model from \autoref{sec:model} with two changes.

First, the \dm's task is to select not necessarily just one~\iitem, but a set of \items\ satisfying some constraints. Below, we denote by $\selectedset = \curlgp{\itemindex \in [N] \mid \selecteditem = 1}$ the \dm's selected set of items. The \dm's choice must satisfy certain \emph{feasibility} constraints. We encode the constraints via a set of feasible sets of items, $\constraints\subseteq 2^{\itemsset}\setminus\emptyset$, and the \dm's choice must be an element of \constraints.
Thus, the process of inspecting and selecting \items\ continues until the \dm\ has selected a feasible set of \items. We call $\constraints$ the problem's \emph{constraints}. We assume $\constraints$ is upward closed, meaning $\selectedset' \supseteq \selectedset \in \constraints$ implies $\selectedset' \in \constraints$.

Second, the \dm's total cost may depend not just on the costs paid to inspect and select boxes, but also on the selected set~$\selectedset$. In particular, a \emph{terminal cost} function $\terminalcost : \constraints \to \reals_{\geq 0}$ exists such that the \dm's total cost is
\[
  \realizedcost
  = \sum_{\itemindex\in\itemsset} \gp{\selecteditem \priceitem + \inspecteditem \inspcostitem}
    + \terminalcost(\selectedset).
\]
We call the pair $(\constraints, \terminalcost)$ a \emph{combinatorial model}, as the constraints and terminal cost function together encode all of the combinatorial structure. The model in \cref{sec:model}, where the \dm\ selects exactly one item, corresponds to $\constraints = \curlgp{\{\itemindex\} \given \itemindex \in \itemsset}$ and $\terminalcost(\selectedset) = 0$.

A combinatorial model $(\constraints,\terminalcost)$ together with the price distribution \distributionitem\ and the inspection costs \inspcostitem\ defines an \emph{instance} of a \emph{combinatorial \optiona\ inspection problem}.  If we additionally impose the constraint that the \dm\ can only select inspected \items, i.e. $\selecteditem \leq \inspecteditem$, we obtain an instance of a combinatorial \emph{\mandator} inspection problem. We let $\optpol(\constraints, \terminalcost)$ (resp., $\mandpol(\constraints, \terminalcost)$) denote the set of policies for \optiona\ (resp., \mandator) inspection problems with combinatorial model $(\constraints, \terminalcost)$.

\paragraph{Organization of the rest of this section}
\Cref{sec:combinatorial:lower-bound} generalizes our lower bound, \cref{theorem:cost-opt}, to combinatorial selection. \Cref{sec:combinatorial:lh} similarly generalizes our upper bound on local hedging, \cref{theorem:lh-pandora}, to combinatorial selection. \Cref{sec:combinatorial:examples} combines local hedging with the \mandator\ inspection policies of \citet{singla2018price} to obtain the first approximation algorithms for combinatorial selection under \optiona\ inspection.

\subsection{Lower bound on optimal cost in the combinatorial setting}
\label{sec:combinatorial:lower-bound}

We now give a lower bound on the optimal expected cost for combinatorial selection. Like the single-item selection case, our lower bound is based on the expected value of a one-shot problem using \opt-surrogate prices. In single-item selection, the one-shot problem is simply taking the minimum of the \opt-surrogate prices. In combinatorial selection, the one-shot problem is instead a one-shot version of the optimization problem induced by the combinatorial model, as captured by the following definition.

\begin{definition}[Surrogate cost]
  \label{definition:realized-cost-surrogate}
  The \emph{optimal \mand-surrogate cost}, denoted~$\realizedcostsurrogate$, is
  \begin{align}\label{eq:surrogate-cost-opt}\tag{\opt-SC}
    \realizedcostsurrogateopt
    = \min_{\selectedset \in \constraints}{} \gp*{\sum_{\itemindex \in \selectedset} \surrogateoptitem + \terminalcost(\selectedset)}.
\end{align}
  Similarly, we define optimal \mand-surrogate and \hedge-surrogate costs, denoted $\realizedcostsurrogate$ and $\realizedcostsurrogatehedge$, by replacing $\surrogateoptitem$ with $\surrogateitem$ and $\surrogatehedgeitem$, respectively.
\end{definition}
To understand the connection between the surrogate cost and the surrogate prices in \autoref{sec:lower-bound}, consider the single-item selection case, in which the constraint set consists of all the singletons and the terminal cost function is identically $0$. In that case, the expression in \autoref{eq:surrogate-cost-opt} reduces to
\[
\realizedcostsurrogateopt = \min_{\itemindex\in\itemsset}\surrogateoptitem,
\]
which is the cost of the one-shot problem from \cref{theorem:cost-opt}.
More generally, \autoref{eq:surrogate-cost-opt} describes the optimal cost of a one-shot problem when prices are given by $\surrogateoptitem$ and the \dm\ incurs no inspection costs, but does incur a terminal cost when selecting a set of items. In that case, the \dm\ will inspect all items--as learning their prices is free--and then select the minimum cost set.


Like in \autoref{sec:lower-bound}, the \opt-surrogate cost provides a lower bound for the combinatorial \optiona\ inspection model.
\restatably
\begin{theorem}[Combinatorial selection lower bound]
  \label{theorem:combinatorial-cost-opt}
  Consider combinatorial selection with model $(\constraints, \terminalcost)$ under \optiona\ inspection. The optimal policy's expected total cost satisfies
  \[
    \min_{\policy \in \optpol(\constraints, \terminalcost)} \E{\policycost} \geq \E{\realizedcostsurrogateopt}.
  \]
\end{theorem}
We prove \cref{theorem:combinatorial-cost-opt} in \cref{appendix:omitted}. The proof follows the same outline as that of \cref{theorem:cost-opt}, namely finding a suitable submartingale, with some extra complications due to the combinatorial nature of the problem.


\paragraph{\Mandator\ inspection}
In light of \autoref{proposition:no-price-info}, it is natural to ask whether an analogous result exists for the combinatorial \mandator\ inspection problem. \Citet[Lemma~2.2]{singla2018price} shows that the \mand-surrogate cost is a lower bound to the optimal cost under \mandator\ inspection, that is,
\[
  \label{eq:surrogate-cost-mand}
  \min_{\policy\in\mandpol} \E{\policycost} \geq \E{\realizedcostsurrogate}.
\]
In other words, in the combinatorial model, surrogate costs provide a benchmark against which to compare different policies, but even in the \mandator\ inspection model, they cease to be a tight benchmark for the optimal cost.

\subsection{Local hedging for combinatorial \optiona\ inspection}\label{sec:combinatorial:lh}

We show in this section how to extend the local hedging policy from single-item selection to combinatorial \optiona\ inspection. Underlying the logic of local hedging in single-item inspection is knowledge of a policy for \mandator\ inspection---the so called Weitzman's rule---that applies across all instances of the single-item \mandator\ inspection. In the combinatorial case, local hedging builds on the existence of such policies for combinatorial \mandator\ inspection problems, such as the frugal algorithms of \citet{singla2018price}.

\begin{definition}[Local hedging in combinatorial \optiona\ inspection]\label{definition:clh}
Fix a $\policy\in\mandpol(\constraints, \terminalcost)$ for combinatorial \mandator\ inspection and a vector $\mathbf{\phedge}=(\phedge_1,\dots,\phedge_N)$ of hedging probabilities.  The \emph{$\policy$ with local $\mathbf{\phedge}$-hedging policy}, denoted $\hedgepi\in\optpol(\constraints, \terminalcost)$ (we leave the $\mathbf{p}$ implicit), is the following two-stage policy:
\* Using the hedging probabilities $\mathbf{\phedge}$, the \dm\ determines the set of inspection and non-inspection \items.
\* The \dm\ then runs policy \policy\ on the resulting combinatorial \mandator\ inspection problem.
\*/
\end{definition}



Our main result below shows, roughly speaking, that if $\pi$ is a $\approxcombo$-approximation for combinatorial selection under \mandator\ inspection, and if all \items\ admit a local \ratio-approximation, then $\hedgepi$ is a $\ratio \approxcombo$-approximation for combinatorial selection under \optiona\ inspection. However, it turns out we need a slightly stronger hypothesis than $\hedgepi$ being a $\approxcombo$-approximation relative to the optimal policy's expected cost, namely the left-hand side of \cref{eq:surrogate-cost-mand}. Instead, we need it to be a $\approxcombo$-approximation relative to the expected \emph{\opt-surrogate cost}, namely the right-hand side of \cref{eq:surrogate-cost-mand}. Fortunately, as we discuss in \cref{sec:combinatorial:examples} below, all the results of \citet{singla2018price} yield approximation algorithms relative to this stricter baseline.

\restatably
\begin{theorem}\label{theorem:lh-combinatorial}
  Let $\policy \in \mandpol(\constraints, \terminalcost)$ be a policy for combinatorial \mandator\ inspection for a given combinatorial model $(\constraints, \terminalcost)$. Suppose that for all hidden price distributions and inspection costs, $\policy$ satisfies $\E{\policycost} \leq \approxcombo \E{\realizedcostsurrogate}$. Then for all hidden price distributions and inspection costs, if all \items\ admit a local \ratio-approximation, then using $\hedgepi$ with the corresponding hedging probabilities yields expected cost bounded by
  \[
    \E{\policycost[\hedgepi]}
    \leq \ratio \approxcombo \E{\realizedcostsurrogateopt}
    \leq \ratio \approxcombo \min_{\policy' \in \optpol(\constraints, \terminalcost)} \E{\policycost[\policy']}.
  \]
  In particular, $\hedgepi$ always yields a $\frac{4}{3} \approxcombo$-approximation for combinatorial selection with \optiona\ inspection.
\end{theorem}

We defer the proof of \cref{theorem:lh-combinatorial} to \cref{appendix:omitted}, giving a brief outline below. There are two main steps. The first step is to show
\[
  \E{\policycost[\hedgepi]} \leq \approxcombo \E{\realizedcostsurrogatehedge}.
\]
This follows immediately from \cref{eq:surrogate-cost-mand}, the assumption that $\pi$ is a $\beta$-approximation relative to the expected \mand-surrogate cost under \mandator\ inspection, and the fact that after labeling items, local hedging transforms the problem into \amandator\ inspection problem. This leaves the second step, which is to show
\[
  \E{\realizedcostsurrogatehedge} \leq \alpha \E{\realizedcostsurrogateopt},
\]
after which \cref{theorem:combinatorial-cost-opt} completes the proof. The second step generalizes the main task in the proof of \cref{theorem:lh}, which is to show
$
  \E{\min_{\itemindex \in \itemsset} \surrogatehedgeitem} \leq \ratio \E{\min_{\itemindex \in \itemsset} \surrogateoptitem}
$.
The single-item selection proof proceeds by, roughly speaking, replacing the \hedge-surrogate prices with \ratio-inflated \opt-surrogate prices one by one. Essentially the same procedure works for combinatorial selection. The main subtlety is that each surrogate price replacement might change the minimizing set of items $\selectedset \in \constraints$ (see \cref{definition:realized-cost-surrogate}). The key idea is to express the surrogate cost in terms of a minimum of one \iitem's surrogate price and a quantity that depends only on other \items' surrogate prices.

\subsection{Extending \mandator\ inspection results of \citet{singla2018price} to \optiona\ inspection}
\label{sec:combinatorial:examples}

\autoref{theorem:lh-combinatorial} leaves open the question of whether policies exists for the \mandator\ combinatorial inspection model that are good approximations of the optimal policy. It turns out that the answer is yes for a number of fundamental problems like matroid basis, set cover, facility location, Steiner-tree, and feedback vertex set. Indeed, \citet[Theorem~1.2]{singla2018price} constructs policies satisfying the precondition of \autoref{theorem:lh-combinatorial} for these and other combinatorial models in the \emph{\mandator\ inspection} setting. Combining these with local hedging yields approximation algorithms for several combinatorial models, as summarized in \autoref{tab:singla}.

\begin{remark}[Baseline of \citeauthor{singla2018price}'s results \citep{singla2018price}]
  While the main result of \citet[Theorem~1.2]{singla2018price} is stated as comparing a policy's performance to that of the optimal \mandator\ inspection policy, namely $\min_{\policy \in \mandpol(\constraints, \terminalcost)} \E{\policycost[\policy]}$, inspecting the proof \citep[Section~3.2]{singla2018price} reveals that all the approximation ratios actually hold relative to the expected \mand-surrogate cost $\E{\realizedcostsurrogate}$, as required by our \cref{theorem:lh-combinatorial}.
\end{remark}


%




\begin{table}
\centering
\caption{Approximation ratios achieved by combining local hedging with policies of \citet{singla2018price} for different \mandator\ inspection problems}
\label{tab:singla}
\begin{tabular}{@{}lll@{}}
\toprule
Problem \& \constraints & Terminal cost $\terminalcost$ & Approximation ratio $\frac{4}{3} \approxcombo$ \\
\midrule
Min-cost matroid basis & $0$ & $\frac{4}{3}$ \\
Min-cost set cover & $0$ & $\min\curlgp[\big]{O(\log n), \frac{4}{3}f}$ \\
Min-cost feedback vertex set & $0$ & $O(\log n)$ \\
Facility location ($\constraints = 2^{\itemsset} \setminus \{\emptyset\}$) & $\sum_{\itemindex\in\itemsset}\min_{s\in\selectedset}d(\itemindex,s)$ & $2.4814$ \\
Steiner tree ($\constraints = 2^{\itemsset}$) & $\textrm{Min-Steiner-Tree}(\itemsset\setminus\selectedset)$ & $4$ \\
\bottomrule
\end{tabular}
\end{table}

As \autoref{tab:singla} illustrates, \autoref{theorem:lh-combinatorial} allows us to tackle \optiona\ inspection combinatorial models in a variety of settings, which have been largely unexplored because of the difficulties introduced by \optiona\ inspection already in the single-item selection case. Consider for instance the uncapacitated \emph{facility location} problem, in which given a graph with vertices $\itemsset$ and edges $E$, the \dm\ must choose a set of locations \selectedset\  at which to open facilities. The \dm\ wants to minimize the cost of the opened facilities, while at the same time minimizing the distance $d:\itemsset \times \itemsset \mapsto\mathbb{R}$ of the facilities to those locations at which no facilities are opened.

Another notable family of combinatorial problems is min-cost matroid basis, of which minimum spanning tree is a special case. Because the optimal algorithm for deterministic min-cost matroid basis is greedy, \citeauthor{singla2018price}'s results yield an algorithm with $\beta = 1$, and hence local hedging achieves an approximation ratio of at most~$\frac{4}{3}$.

We refer the reader to \citet{singla2018price} for detailed descriptions of the problems in \cref{tab:singla}. We emphasize that the power of our approach is not that it addresses any particular combinatorial problem, but rather its \emph{compositionality}: local hedging naturally combines with algorithms for \mandator\ inspection problems and translates both the algorithms and their performance guarantees to \mandator\ inspection problems, provided the performance guarantees are relative to expected \mand-surrogate cost.

\section{Conclusion and Discussion}\label{sec:conclusions}

In this work, we introduce a new approach to approximately solving Pandora's box problems with \optiona\ inspection. Our approach, \emph{local hedging}, maintains the simplicity and compositionality of the elegant policies available for Pandora's box problems with \mandator\ inspection. One can view local hedging as a randomized reduction that turns \optiona\ inspection problems into \mandator\ inspection problems. The result is the first approximation algorithms for combinatorial Pandora's box problems under \optiona\ inspection.

We believe the local hedging technique has potential to be used beyond the setting of this paper. In the rest of this section, we outline the possibilities and obstacles to using local hedging in two additional settings. \Cref{sec:conclusions:rewards} discusses \emph{reward maximization} Pandora's box problems, in contrast to the cost minimization setting we focus on. \Cref{sec:conclusions:superprocesses} is \emph{Markovian bandit superprocesses}, which significantly generalize Pandora's box models.

\subsection{Reward maximization}
\label{sec:conclusions:rewards}

In Pandora's box problems in the reward maximization setting, instead of each \iitem\ having a hidden price, each \iitem\ has a \emph{hidden reward}. The objective is to maximize expected reward of selected items minus inspection costs (possibly plus a terminal reward in the combinatorial setting).

The core Pandora's box definitions easily translate between the cost and reward settings. The rule of thumb is that one can recover definitions for rewards by interpreting them as negative costs, and vice versa. We do this, for instance, when translating \citeauthor{doval2018whether}'s characterization of the \subproblem\ \citep[Proposition~0]{doval2018whether} from rewards to costs. The definitions of local hedging and local approximation can similarly be translated from costs to rewards. More generally, if one allows negative hidden prices or hidden rewards, then the two settings are essentially the same.

While allowing negative costs does not fundamentally alter any of our definitions, it does have an impact on our main local approximation result. Recall that \cref{theorem:lh} states that any \iitem\ admits a local $\frac{4}{3}$-approximation. However, the proof of the $\frac{4}{3}$ upper bound relies on the fact that an \iitem's reservation price is bounded below, specifically $\rv \geq \inspcost$. When one allows negative costs, there is no longer such a lower bound, and an \iitem's reservation price may be arbitrarily negative.

Translating the above discussion to the rewards setting, an \iitem's \emph{reservation value} (namely its negative reservation price) may be an arbitrarily large relative to its mean value and inspection cost. Consequently, when translating the proof of \cref{theorem:lh} to the rewards setting, one finds that the worst possible approximation ratio is only~$\frac{1}{2}$. This is a disappointing result, because \citet{beyhaghi2019pandora} point out that the following trivial randomized committing policy is a $\frac{1}{2}$-approximation in the rewards case: with equal probability, commit to either never inspecting any \items\ or never selecting any \items\ without inspection.\footnote{%
  Curiously, we are not aware of a similar trivial randomized algorithm for the costs setting.}
So while local hedging may still yield a superior result for \items\ that admit local $\ratio$-approximations for $\ratio > \frac{1}{2}$, it appears that simple local hedging alone cannot replicate the best simple approximation algorithms for the rewards setting, which achieve a $\frac{4}{5}$-approximation \citep{guha2008information}.

It is an interesting open question whether the definition of local approximation can be altered in a way that enables improved guarantees in the rewards setting. One idea would be to combine the multiplicative suboptimality factor currently considered with an additive suboptimality gap, which we suspect could rule out the worst-case examples that admit only local $\frac{1}{2}$-approximation. A follow-up to this work \citep{chawla2024combinatorial} has shown this is in principle possible, resulting in a $0.582$-approximation for combinatorial selection with matroid constraints.

\subsection{Bandit superprocesses}
\label{sec:conclusions:superprocesses}

As \citet{doval2018whether} demonstrates through multiple examples, the core reason why Pandora's box problems are harder under \optiona\ inspection than \mandator\ inspection is that \optiona\ inspection gives the \dm\ two possible actions to take on each box, which makes the problem a \emph{Markovian bandit superprocess}. We see potential for applying local hedging to certain cases of this more general superprocess setting.

We begin with some background. Roughly speaking, a \emph{Markovian bandit process} is a type of Markov decision process in which the \dm\ must at each time step choose between advancing one of multiple independent Markov chains \citep{gittins2011multi}. The traditional setting in which bandit processes are studied is infinite-horizon discounted reward problems \citep{weber1992gittins}, but variants exist for unconstrained-but-finite undiscounted problems, like Pandora's box problems. One example is the model of \citet{dumitriu2003playing}, which, aside from the assumption of discrete state spaces, is a multistage generalization of single-item selection in the cost minimization setting.

Many definitions and results available for Pandora's box problems under \mandator\ inspection translate to bandit processes. Most importantly, reservation prices (or reservation values) are a special case of \emph{Gittins indices}, which form the basis of optimal index policies for many varieties of bandit process \citep{gittins2011multi}. Gittins indices share the simplicity and compositionality of reservation prices. For instance, \citet{gupta2019markovian} generalize the combinatorial selection results of \citet{singla2018price} to a more general model resembling that of \citet{dumitriu2003playing}. The core reason why these results work is that even though the \dm\ has many choices of Markov chains to advance at each time step, only one action, namely ``advance'', is possible within each Markov chain.

Markovian bandit \emph{superprocesses} generalize bandit processes by replacing the independent Markov chains with independent Markov \emph{decision processes}. That is, now each of the independent processes may present the \dm\ with multiple possible actions. In Pandora's box problems, uninspected \items\ under \optiona\ inspection are an example of this, as they allow either inspection or immediate selection. While one can define Gittins indices for bandit superprocesses, unlike the simple bandit process case, they generally do not yield optimal policies. The only known exception is when each Markov decision process satisfies a condition known as \emph{Whittle's condition} \citep{whittle1980multi, glazebrook1982sufficient}. Roughly speaking, a Markov decision process satisfies Whittle's condition if one would be willing to commit to a state-to-action mapping ahead of time, then always use that to determine which action to play, regardless of the states of other Markov decision processes in the bandit superprocess.

Our notion of local $\alpha$-approximation can be viewed as a \emph{novel relaxation of Whittle's condition}. Specifically, in our Pandora's box setting, satisfying Whittle's condition amounts to admitting a local $1$-approximation. One can easily generalize the definition of local approximation to the more general superprocess setting, although there are some subtleties to work out about what types of randomness should be allowed in the state-to-action mapping. A follow-up to this work \citep{chawla2024combinatorial} has shown that analogues of our \cref{theorem:lh-pandora, theorem:lh-combinatorial} hold in a more general model that captures many variants on Pandora's box beyond \optiona\ inspection, with essentially the same proof outline.

Because local approximations relax Whittle's condition, the set of Markov decision processes that admit local $\alpha$-approximations for values of~$\alpha$ reasonably close to~$1$ is likely to be much richer than those satisfying Whittle's condition. For instance, in Pandora's box with \optiona\ inspection, the only \items\ that satisfy Whittle's condition are those for which inspection is never worthwhile \citep{doval2018whether}. In contrast, all \items\ admit local $\frac{4}{3}$-approximations, and a follow-up to this work \citep{chawla2024combinatorial} has identified other variants of Pandora's box that always admit good local approximations. We therefore believe that local hedging and local approximations give a promising new angle of attack for deriving approximation algorithms for bandit superprocesses.

\bibliographystyle{ACM-Reference-Format}
\bibliography{pandora}


\begin{thebibliography}{29}


\ifx \showCODEN    \undefined \def \showCODEN     #1{\unskip}     \fi
\ifx \showDOI      \undefined \def \showDOI       #1{#1}\fi
\ifx \showISBNx    \undefined \def \showISBNx     #1{\unskip}     \fi
\ifx \showISBNxiii \undefined \def \showISBNxiii  #1{\unskip}     \fi
\ifx \showISSN     \undefined \def \showISSN      #1{\unskip}     \fi
\ifx \showLCCN     \undefined \def \showLCCN      #1{\unskip}     \fi
\ifx \shownote     \undefined \def \shownote      #1{#1}          \fi
\ifx \showarticletitle \undefined \def \showarticletitle #1{#1}   \fi
\ifx \showURL      \undefined \def \showURL       {\relax}        \fi
\providecommand\bibfield[2]{#2}
\providecommand\bibinfo[2]{#2}
\providecommand\natexlab[1]{#1}
\providecommand\showeprint[2][]{arXiv:#2}

\bibitem[Aminian et~al\mbox{.}(2023)]%
        {aminian2023fair}
\bibfield{author}{\bibinfo{person}{Mohammad~Reza Aminian},
  \bibinfo{person}{Vahideh Manshadi}, {and} \bibinfo{person}{Rad Niazadeh}.}
  \bibinfo{year}{2023}\natexlab{}.
\newblock \showarticletitle{Fair Markovian Search}.
\newblock \bibinfo{journal}{\emph{Available at SSRN 4347447}}
  (\bibinfo{year}{2023}).
\newblock


\bibitem[Aouad et~al\mbox{.}(2020)]%
        {aouad2020pandora}
\bibfield{author}{\bibinfo{person}{Ali Aouad}, \bibinfo{person}{Jingwei Ji},
  {and} \bibinfo{person}{Yaron Shaposhnik}.} \bibinfo{year}{2020}\natexlab{}.
\newblock \showarticletitle{The Pandora's Box Problem with Sequential
  Inspections}.
\newblock \bibinfo{journal}{\emph{Available at SSRN 3726167}}
  (\bibinfo{year}{2020}).
\newblock


\bibitem[Armstrong(2017)]%
        {armstrong2017ordered}
\bibfield{author}{\bibinfo{person}{Mark Armstrong}.}
  \bibinfo{year}{2017}\natexlab{}.
\newblock \showarticletitle{Ordered consumer search}.
\newblock \bibinfo{journal}{\emph{Journal of the European Economic
  Association}} \bibinfo{volume}{15}, \bibinfo{number}{5}
  (\bibinfo{year}{2017}), \bibinfo{pages}{989--1024}.
\newblock


\bibitem[Beyhaghi and Cai(2023a)]%
        {beyhaghi2023pandora}
\bibfield{author}{\bibinfo{person}{Hedyeh Beyhaghi} {and}
  \bibinfo{person}{Linda Cai}.} \bibinfo{year}{2023}\natexlab{a}.
\newblock \showarticletitle{Pandora’s problem with nonobligatory inspection:
  Optimal structure and a PTAS}. In \bibinfo{booktitle}{\emph{Proceedings of
  the 55th Annual ACM Symposium on Theory of Computing}}.
  \bibinfo{pages}{803--816}.
\newblock


\bibitem[Beyhaghi and Cai(2023b)]%
        {beyhaghi2023recent}
\bibfield{author}{\bibinfo{person}{Hedyeh Beyhaghi} {and}
  \bibinfo{person}{Linda Cai}.} \bibinfo{year}{2023}\natexlab{b}.
\newblock \showarticletitle{Recent developments in pandora's box problem:
  Variants and applications}.
\newblock \bibinfo{journal}{\emph{ACM SIGecom Exchanges}} \bibinfo{volume}{21},
  \bibinfo{number}{1} (\bibinfo{year}{2023}), \bibinfo{pages}{20--34}.
\newblock


\bibitem[Beyhaghi and Kleinberg(2019)]%
        {beyhaghi2019pandora}
\bibfield{author}{\bibinfo{person}{Hedyeh Beyhaghi} {and}
  \bibinfo{person}{Robert Kleinberg}.} \bibinfo{year}{2019}\natexlab{}.
\newblock \showarticletitle{Pandora's Problem with Nonobligatory Inspection}.
\newblock \bibinfo{journal}{\emph{arXiv preprint arXiv:1905.01428}}
  (\bibinfo{year}{2019}).
\newblock


\bibitem[Bhaskara et~al\mbox{.}(2022)]%
        {bhaskara2022online}
\bibfield{author}{\bibinfo{person}{Aditya Bhaskara}, \bibinfo{person}{Sreenivas
  Gollapudi}, \bibinfo{person}{Sungjin Im}, \bibinfo{person}{Kostas Kollias},
  {and} \bibinfo{person}{Kamesh Munagala}.} \bibinfo{year}{2022}\natexlab{}.
\newblock \showarticletitle{Online Learning and Bandits with Queried Hints}.
\newblock \bibinfo{journal}{\emph{arXiv preprint arXiv:2211.02703}}
  (\bibinfo{year}{2022}).
\newblock


\bibitem[Boodaghians et~al\mbox{.}(2020)]%
        {boodaghians2020pandora}
\bibfield{author}{\bibinfo{person}{Shant Boodaghians},
  \bibinfo{person}{Federico Fusco}, \bibinfo{person}{Philip Lazos}, {and}
  \bibinfo{person}{Stefano Leonardi}.} \bibinfo{year}{2020}\natexlab{}.
\newblock \showarticletitle{Pandora's box problem with order constraints}. In
  \bibinfo{booktitle}{\emph{Proceedings of the 21st ACM Conference on Economics
  and Computation}}. \bibinfo{pages}{439--458}.
\newblock


\bibitem[Brown and Smith(2013)]%
        {brown2013optimal}
\bibfield{author}{\bibinfo{person}{David~B Brown} {and}
  \bibinfo{person}{James~E Smith}.} \bibinfo{year}{2013}\natexlab{}.
\newblock \showarticletitle{Optimal sequential exploration: Bandits,
  clairvoyants, and wildcats}.
\newblock \bibinfo{journal}{\emph{Operations research}} \bibinfo{volume}{61},
  \bibinfo{number}{3} (\bibinfo{year}{2013}), \bibinfo{pages}{644--665}.
\newblock


\bibitem[Chawla et~al\mbox{.}(2024)]%
        {chawla2024combinatorial}
\bibfield{author}{\bibinfo{person}{Shuchi Chawla}, \bibinfo{person}{Dimitris
  Christou}, \bibinfo{person}{Amit Harlev}, {and} \bibinfo{person}{Ziv
  Scully}.} \bibinfo{year}{2024}\natexlab{}.
\newblock \showarticletitle{Combinatorial Selection with Costly Information}.
\newblock \bibinfo{journal}{\emph{arXiv preprint arXiv:2412.03860}}
  (\bibinfo{year}{2024}).
\newblock


\bibitem[Chawla et~al\mbox{.}(2020)]%
        {chawla2020pandora}
\bibfield{author}{\bibinfo{person}{Shuchi Chawla}, \bibinfo{person}{Evangelia
  Gergatsouli}, \bibinfo{person}{Yifeng Teng}, \bibinfo{person}{Christos
  Tzamos}, {and} \bibinfo{person}{Ruimin Zhang}.}
  \bibinfo{year}{2020}\natexlab{}.
\newblock \showarticletitle{Pandora's box with correlations: Learning and
  approximation}. In \bibinfo{booktitle}{\emph{2020 IEEE 61st Annual Symposium
  on Foundations of Computer Science (FOCS)}}. IEEE,
  \bibinfo{pages}{1214--1225}.
\newblock


\bibitem[Derakhshan et~al\mbox{.}(2022)]%
        {derakhshan2022product}
\bibfield{author}{\bibinfo{person}{Mahsa Derakhshan}, \bibinfo{person}{Negin
  Golrezaei}, \bibinfo{person}{Vahideh Manshadi}, {and} \bibinfo{person}{Vahab
  Mirrokni}.} \bibinfo{year}{2022}\natexlab{}.
\newblock \showarticletitle{Product ranking on online platforms}.
\newblock \bibinfo{journal}{\emph{Management Science}} \bibinfo{volume}{68},
  \bibinfo{number}{6} (\bibinfo{year}{2022}), \bibinfo{pages}{4024--4041}.
\newblock


\bibitem[Doval(2018)]%
        {doval2018whether}
\bibfield{author}{\bibinfo{person}{Laura Doval}.}
  \bibinfo{year}{2018}\natexlab{}.
\newblock \showarticletitle{Whether or not to open Pandora's box}.
\newblock \bibinfo{journal}{\emph{Journal of Economic Theory}}
  \bibinfo{volume}{175} (\bibinfo{year}{2018}), \bibinfo{pages}{127--158}.
\newblock


\bibitem[Dumitriu et~al\mbox{.}(2003)]%
        {dumitriu2003playing}
\bibfield{author}{\bibinfo{person}{Ioana Dumitriu}, \bibinfo{person}{Prasad
  Tetali}, {and} \bibinfo{person}{Peter Winkler}.}
  \bibinfo{year}{2003}\natexlab{}.
\newblock \showarticletitle{On playing golf with two balls}.
\newblock \bibinfo{journal}{\emph{SIAM Journal on Discrete Mathematics}}
  \bibinfo{volume}{16}, \bibinfo{number}{4} (\bibinfo{year}{2003}),
  \bibinfo{pages}{604--615}.
\newblock


\bibitem[Fu et~al\mbox{.}(2023)]%
        {fu2023pandora}
\bibfield{author}{\bibinfo{person}{Hu Fu}, \bibinfo{person}{Jiawei Li}, {and}
  \bibinfo{person}{Daogao Liu}.} \bibinfo{year}{2023}\natexlab{}.
\newblock \showarticletitle{Pandora box problem with nonobligatory inspection:
  Hardness and approximation scheme}. In \bibinfo{booktitle}{\emph{Proceedings
  of the 55th Annual ACM Symposium on Theory of Computing}}.
  \bibinfo{pages}{789--802}.
\newblock


\bibitem[Gergatsouli and Tzamos(2022)]%
        {gergatsouli2022online}
\bibfield{author}{\bibinfo{person}{Evangelia Gergatsouli} {and}
  \bibinfo{person}{Christos Tzamos}.} \bibinfo{year}{2022}\natexlab{}.
\newblock \showarticletitle{Online learning for min sum set cover and
  pandora’s box}. In \bibinfo{booktitle}{\emph{International Conference on
  Machine Learning}}. PMLR, \bibinfo{pages}{7382--7403}.
\newblock


\bibitem[Gergatsouli and Tzamos(2023)]%
        {gergatsouli2023weitzman}
\bibfield{author}{\bibinfo{person}{Evangelia Gergatsouli} {and}
  \bibinfo{person}{Christos Tzamos}.} \bibinfo{year}{2023}\natexlab{}.
\newblock \showarticletitle{Weitzman's Rule for Pandora's Box with
  Correlations}.
\newblock \bibinfo{journal}{\emph{arXiv preprint arXiv:2301.13534}}
  (\bibinfo{year}{2023}).
\newblock


\bibitem[Gittins et~al\mbox{.}(2011)]%
        {gittins2011multi}
\bibfield{author}{\bibinfo{person}{John Gittins}, \bibinfo{person}{Kevin
  Glazebrook}, {and} \bibinfo{person}{Richard Weber}.}
  \bibinfo{year}{2011}\natexlab{}.
\newblock \bibinfo{booktitle}{\emph{Multi-armed bandit allocation indices}}.
\newblock \bibinfo{publisher}{John Wiley \& Sons}.
\newblock


\bibitem[Glazebrook(1982)]%
        {glazebrook1982sufficient}
\bibfield{author}{\bibinfo{person}{Kevin~D Glazebrook}.}
  \bibinfo{year}{1982}\natexlab{}.
\newblock \showarticletitle{On a sufficient condition for superprocesses due to
  Whittle}.
\newblock \bibinfo{journal}{\emph{Journal of Applied Probability}}
  \bibinfo{volume}{19}, \bibinfo{number}{1} (\bibinfo{year}{1982}),
  \bibinfo{pages}{99--110}.
\newblock


\bibitem[Guha et~al\mbox{.}(2008)]%
        {guha2008information}
\bibfield{author}{\bibinfo{person}{Sudipto Guha}, \bibinfo{person}{Kamesh
  Munagala}, {and} \bibinfo{person}{Saswati Sarkar}.}
  \bibinfo{year}{2008}\natexlab{}.
\newblock \showarticletitle{Information acquisition and exploitation in
  multichannel wireless networks}.
\newblock \bibinfo{journal}{\emph{arXiv preprint arXiv:0804.1724}}
  (\bibinfo{year}{2008}).
\newblock


\bibitem[Gupta et~al\mbox{.}(2019)]%
        {gupta2019markovian}
\bibfield{author}{\bibinfo{person}{Anupam Gupta}, \bibinfo{person}{Haotian
  Jiang}, \bibinfo{person}{Ziv Scully}, {and} \bibinfo{person}{Sahil Singla}.}
  \bibinfo{year}{2019}\natexlab{}.
\newblock \showarticletitle{The markovian price of information}. In
  \bibinfo{booktitle}{\emph{Integer Programming and Combinatorial Optimization:
  20th International Conference, IPCO 2019, Ann Arbor, MI, USA, May 22-24,
  2019, Proceedings 20}}. Springer, \bibinfo{pages}{233--246}.
\newblock


\bibitem[Hoefer et~al\mbox{.}(2021)]%
        {hoefer2021stochastic}
\bibfield{author}{\bibinfo{person}{Martin Hoefer}, \bibinfo{person}{Kevin
  Schewior}, {and} \bibinfo{person}{Daniel Schmand}.}
  \bibinfo{year}{2021}\natexlab{}.
\newblock \showarticletitle{Stochastic Probing with Increasing Precision}. In
  \bibinfo{booktitle}{\emph{IJCAI}}. \bibinfo{pages}{4069--4075}.
\newblock


\bibitem[Klabjan et~al\mbox{.}(2014)]%
        {klabjan2014attributes}
\bibfield{author}{\bibinfo{person}{Diego Klabjan}, \bibinfo{person}{Wojciech
  Olszewski}, {and} \bibinfo{person}{Asher Wolinsky}.}
  \bibinfo{year}{2014}\natexlab{}.
\newblock \showarticletitle{Attributes}.
\newblock \bibinfo{journal}{\emph{Games and Economic Behavior}}
  \bibinfo{volume}{88} (\bibinfo{year}{2014}), \bibinfo{pages}{190--206}.
\newblock


\bibitem[Kleinberg et~al\mbox{.}(2016)]%
        {kleinberg2016descending}
\bibfield{author}{\bibinfo{person}{Robert Kleinberg}, \bibinfo{person}{Bo
  Waggoner}, {and} \bibinfo{person}{E~Glen Weyl}.}
  \bibinfo{year}{2016}\natexlab{}.
\newblock \showarticletitle{Descending price optimally coordinates search}.
\newblock \bibinfo{journal}{\emph{arXiv preprint arXiv:1603.07682}}
  (\bibinfo{year}{2016}).
\newblock


\bibitem[Singla(2018)]%
        {singla2018price}
\bibfield{author}{\bibinfo{person}{Sahil Singla}.}
  \bibinfo{year}{2018}\natexlab{}.
\newblock \showarticletitle{The price of information in combinatorial
  optimization}. In \bibinfo{booktitle}{\emph{Proceedings of the twenty-ninth
  annual ACM-SIAM symposium on discrete algorithms}}. SIAM,
  \bibinfo{pages}{2523--2532}.
\newblock


\bibitem[Ursu et~al\mbox{.}(2023)]%
        {ursu2023sequential}
\bibfield{author}{\bibinfo{person}{Raluca Ursu}, \bibinfo{person}{Stephan
  Seiler}, {and} \bibinfo{person}{Elisabeth Honka}.}
  \bibinfo{year}{2023}\natexlab{}.
\newblock \showarticletitle{The Sequential Search Model: A Framework for
  Empirical Research}.
\newblock  (\bibinfo{year}{2023}).
\newblock


\bibitem[Weber(1992)]%
        {weber1992gittins}
\bibfield{author}{\bibinfo{person}{Richard Weber}.}
  \bibinfo{year}{1992}\natexlab{}.
\newblock \showarticletitle{On the Gittins index for multiarmed bandits}.
\newblock \bibinfo{journal}{\emph{The Annals of Applied Probability}}
  (\bibinfo{year}{1992}), \bibinfo{pages}{1024--1033}.
\newblock


\bibitem[Weitzman(1979)]%
        {weitzman1979optimal}
\bibfield{author}{\bibinfo{person}{Martin~L Weitzman}.}
  \bibinfo{year}{1979}\natexlab{}.
\newblock \showarticletitle{Optimal search for the best alternative}.
\newblock \bibinfo{journal}{\emph{Econometrica: Journal of the Econometric
  Society}} (\bibinfo{year}{1979}), \bibinfo{pages}{641--654}.
\newblock


\bibitem[Whittle(1980)]%
        {whittle1980multi}
\bibfield{author}{\bibinfo{person}{Peter Whittle}.}
  \bibinfo{year}{1980}\natexlab{}.
\newblock \showarticletitle{Multi-armed bandits and the Gittins index}.
\newblock \bibinfo{journal}{\emph{Journal of the Royal Statistical Society:
  Series B (Methodological)}} \bibinfo{volume}{42}, \bibinfo{number}{2}
  (\bibinfo{year}{1980}), \bibinfo{pages}{143--149}.
\newblock


\end{thebibliography}

\newpage
\appendix

\section{Deferred proofs}\label{appendix:omitted}

\restate*{lemma:min-r-surrogate}
\begin{proof}
  \label{proof:min-r-surrogate-opt}

  The \mandator\ inspection statement is standard (see, e.g., \citealp[Lemma 1]{kleinberg2016descending}), so we turn immediately to the \optiona\ inspection statement. We consider two cases:

  \begin{case}[$\rv < \bv$]
    Expanding both sides using \cref{definition:surrogate, definition:surrogate-optional}, we aim to show
    \[
      \label{eq:surrogate-opt-vs-surrogate:goal}
      \E[\big]{\min\{\max\{\price, \rv\}, \bv, r\}}
      = \min\curlgp[\big]{\E[\big]{\min\{\max\{\price, \rv\}, r\}}, \mean}.
    \]
    We clearly have
    \[
      \label{eq:surrogate-opt-vs-surrogate:r_small}
      \E[\big]{\min\{\max\{\price, \rv\}, \bv, r\}}
      \leq \E[\big]{\min\{\max\{\price, \rv\}, r\}},
    \]
    and using \cref{definition:surrogate, definition:surrogate-optional} and the $\rv < \bv$ assumption, we compute
    \[
      \E[\big]{\min\{\max\{\price, \rv\}, \bv, r\}}
      \label{eq:surrogate-opt-vs-surrogate:r_large}
      &\leq \E[\big]{\min\{\max\{\price, \rv\}, \bv\}} \\
      &= \meanitem + \E{(\rv - \price)^+} - \E{(\price - \bv)^+} \\
      &= \mean + \inspcost - \inspcost= \mean.
    \]
    One of \cref{eq:surrogate-opt-vs-surrogate:r_small} or \cref{eq:surrogate-opt-vs-surrogate:r_large} holds with equality (because $r \leq \bv$ or $r \geq \bv$), implying \cref{eq:surrogate-opt-vs-surrogate:goal}.
  \end{case}

  \begin{case}[$\rv \geq \bv$]
    Expanding both sides using \cref{definition:surrogate-optional}, we aim to show
    \[
      \min\{\mean, r\} = \min\curlgp[\big]{\mean, \E{\min\{\surrogate, r\}}}.
    \]
    The left-hand side is greater than or equal to the right-hand side, so it remains only to show the reverse inequality. By \cref{definition:surrogate}, we have $\rv \leq \surrogate$ with probability~$1$, so it suffices to show $\mean \leq \rv$. This holds due to the $\rv \geq \bv$ assumption, \cref{eq:reserve, eq:backup}, and the following computation:
    \[
      \mean = \E{\price}
      &= \rv + \E{(\price - \rv)^+} - \E{(\rv - \price)^+} \\
      &\leq \rv + \E{(\price - \bv)^+} - \E{(\rv - \price)^+} \\
      &= \rv + \inspcost - \inspcost = \rv.
      \qedhere
    \]
  \end{case}
\end{proof}

\restate*{theorem:cost-opt}

\begin{proof}
  \label{proof:cost-opt}
  Consider an arbitrary policy~$\policy$ for the \dm. After $t$ rounds, the policy has inspected some \items\ and possibly selected one. Let\footnote{%
    Below, we allow $t$ to be greater than the number of rounds that $\policy$ takes to select an \iitem. We follow the convention that the process simply remains static after an \iitem\ is selected.}
  \[
    \realizedcost(t) &= \textnormal{total inspection and selection cost paid by $\policy$ during $\{0, \dots, t - 1\}$,} \\
    \surrogatebaseitem(t) &=
    \begin{cases}
      0 & \textnormal{if any \iitem\ is selected by $\policy$ during $\{0, \dots, t - 1\}$} \\
      \priceitem & \textnormal{if $\itemindex$ is inspected, but not selected, by $\policy$ during $\{0, \dots, t - 1\}$} \\
      \surrogateoptitem & \textnormal{otherwise,}
    \end{cases} \\
    \filtration(t) &= \textnormal{information $\policy$ gains from inspections during $\{0, \dots, t - 1\}$,} \\
    K(t) &= \realizedcost(t) + \E*{\min_{\itemindex \in \itemsset} \surrogatebaseitem(t) \given \filtration(t)}.
  \]
  After $\numitems + 1$ rounds, the \dm\ will have selected an \iitem\ and the process will have terminated, so $K(\numitems + 1) = \policycost$. We also have $K(0) = \E{\min_{\itemindex \in \itemsset} \surrogateopt}$, so it suffices to show that $\{K(t)\}_t$ is a submartingale with respect to $\{\filtration(t)\}_t$.

  Below, to reduce clutter, we abbreviate $\E{\cdot \given \filtration(t)}$ to $\E_t{\cdot}$.

  We aim to show $K(t) \leq \E_t{K(t + 1)}$. To do so, we consider each action the \dm\ might take on each \iitem\ $m \in \itemsset$. For each action, we write the difference $\E_t{K(t + 1)} - K(t)$ in terms of the quantity
  \[
    \SurMin(t) = \min_{\itemindex \in \itemsset \setminus \{m\}} \surrogatebaseitem(t).
  \]
  The fact that $\SurMin(t)$ depends only on $\surrogatebaseitem(t)$ for $\itemindex \neq m$ implies the following observations.
  \*[(O1)] If the \dm\ takes an action on item~$m$ at time~$t$, then $\SurMin(t) = \SurMin(t + 1)$.
  \* $\SurMin(t)$ is conditionally independent of $\surrogatebaseitem[m](t)$ given $\filtration(t)$.
  \*/

  With the appropriate notation and the above observations in hand, we can show the difference $\E_t{K(t + 1)} - K(t)$ is nonnegative no matter which action the \dm\ takes.
  \* If \iitem~$m$ is closed and the \dm\ inspects it, then
  \[
    \E_t{K(t + 1)} - K(t)
    &= -\E_t{\min\{\surrogateoptitem[m], \SurMin(t)\}} + \inspcostitem[m] + \E_t{\min\{\priceitem[m], \SurMin(t + 1)\}} \\
    &= -\E_t{\min\{\surrogateoptitem[m], \SurMin(t)\}} + \inspcostitem[m] + \E_t{\min\{\priceitem[m], \SurMin(t)\}}.
  \]
  The second equality follows by~(O1), and its right-hand side is nonnegative by (O2) and \cref{lemma:min-r-surrogate}.
  \* If \iitem~$m$ is closed and the \dm\ selects it (without inspection), then
  \[
    \E_t{K(t + 1)} - K(t)
    &= -\E_t{\min\{\surrogateoptitem[m], \SurMin(t + 1)\}} + \meanitem[m] \\
    &= -\E_t{\min\{\surrogateoptitem[m], \SurMin(t)\}} + \meanitem[m].
  \]
  Again, the second equality follows by~(O1), and its right-hand side is nonnegative by (O2) and \cref{lemma:min-r-surrogate}.
  \* If \iitem~$m$ is open and the \dm\ selects it, then
  \[
    \E_t{K(t + 1)} - K(t)
    = -\E_t{\min\{\priceitem[m], \SurMin(t)\}} + \priceitem[m].
  \]
  This is nonnegative because $\priceitem[m]$ is known to the \dm\ at time~$t$, so $\priceitem[m] = \E_t{\priceitem[m]}$.
  \qedhere
  \*/
\end{proof}

\restate*{theorem:combinatorial-cost-opt}

\begin{proof}
  This proof follows essentially the same steps as the proof of \cref{theorem:cost-opt} above, but with some extra complications to handle the combinatorial aspect.

  Consider an arbitrary policy~$\policy$ for the \dm. After $t$ rounds, the policy has inspected some \items\ and possibly selected one. Let\footnote{%
    Below, we allow $t$ to be greater than the number of rounds that $\policy$ takes to select an admissible set of \items. We follow the convention that the process simply remains static after such a set is selected.}
  \[
    \realizedcost(t) &= \textnormal{total inspection and selection cost paid by $\policy$ during $\{0, \dots, t - 1\}$,} \\
    \surrogatebaseitem(t) &=
    \begin{cases}
      0 & \textnormal{if any \iitem\ is selected by $\policy$ during $\{0, \dots, t - 1\}$} \\
      \priceitem & \textnormal{if $\itemindex$ is inspected, but not selected, by $\policy$ during $\{0, \dots, t - 1\}$} \\
      \surrogateoptitem & \textnormal{otherwise,}
    \end{cases} \\
    \filtration(t) &= \textnormal{information $\policy$ gains from inspections during $\{0, \dots, t - 1\}$,} \\
    \constraints(t) &= \curlgp[\big]{\selectedset \in \constraints \given \textnormal{$\selectedset$ contains all \items\ selected by $\pi$ during $\{0, \dots, t - 1\}$}} \\
    \realizedcostsurrogatebase(t) &= \min_{\selectedset \in \constraints(t)}{} \gp*{\sum_{\itemindex \in \selectedset} \surrogatebaseitem(t) + \terminalcost(\selectedset)}, \\
    K(t) &= \realizedcost(t) + \E{\realizedcostsurrogatebase(t) \given \filtration(t)}.
  \]
  After $\numitems + 1$ rounds, the \dm\ will have selected an \iitem\ and the process will have terminated, so $K(\numitems + 1) = \policycost$. We also have $K(0) = \E{\min_{\itemindex \in \itemsset} \surrogateopt}$, so it suffices to show that $\{K(t)\}_t$ is a submartingale with respect to $\{\filtration(t)\}_t$.

  Below, to reduce clutter, we abbreviate $\E{\cdot\! \given \filtration(t)}$ to $\E_t{\cdot}$.

  We aim to show $K(t) \leq \E_t{K(t + 1)}$. To do so, we consider each action the \dm\ might take on each \iitem\ $m \in \itemsset$. For each action, write the difference $\E_t{K(t + 1)} - K(t)$ in terms of the quantities
  \[
    \SurMinWithout(t) &= \min_{\selectedset \in \constraints(t) \,|\, m \not\in \selectedset}{} \gp*{
      \sum_{\itemindex \in \selectedset} \surrogatebaseitem(t)
      + \terminalcost(\selectedset)
    },
    \\
    \SurMinWith(t) &= \min_{\selectedset \in \constraints(t) \,|\, m \in \selectedset}{} \gp[\Bigg]{
      \mkern0.75mu \smashoperator[r]{\sum_{\itemindex \in \selectedset \setminus \{m\}}}
      \surrogatebaseitem(t) + \terminalcost(\selectedset)
    },
    \\
    \SurMin(t) &= \SurMinWithout(t) - \SurMinWith(t).
  \]
  We use the convention that a minimum over an empty set is~$\infty$. That is, if all sets in $\constraints(t)$ contain $m$, then $\SurMinWithout(t) = \infty$, and similarly if none contain~$m$, then $\SurMinWith(t) = \infty$.\footnote{%
    Provided the policy~$\policy$ never fails to select a feasible set of items, it will never be the case that both $\SurMinWithout(t)$ and $\SurMinWith(t)$ are infinite.}

  These quantities above give a useful decomposition of $\realizedcostsurrogatebase(t)$, namely that for any \iitem~$m$,
  \[
    \realizedcostsurrogatebase(t)
    &= \min\{\SurMinWithout(t), \SurMinWith(t) + \surrogatebaseitem[m](t)\} \\
    &= \SurMinWith(t) + \min\{\surrogatebaseitem[m](t), \SurMin(t)\}.
  \]
  The fact that $\SurMinWith(t)$ and $\SurMin(t)$ depend only on $\surrogatebaseitem(t)$ for $\itemindex \neq m$ implies the following observations.
  \*[(O1)] If the \dm\ takes an action on item~$m$ at time~$t$, then $\SurMinWithout(t) \leq \SurMinWithout(t + 1)$ and $\SurMinWith(t) = \SurMinWith(t + 1)$, so $\SurMin(t) \leq \SurMin(t + 1)$.
  \* $\SurMin(t)$ is conditionally independent of $\surrogatebaseitem[m](t)$ given $\filtration(t)$.
  \*/

  With the appropriate notation and the above observations in hand, we can show the difference $\E_t{K(t + 1)} - K(t)$ is nonnegative no matter which action the \dm\ takes.
  \* If \iitem~$m$ is closed and the \dm\ inspects it, then
  \[
    \MoveEqLeft
    \E_t{K(t + 1)} - K(t) \\*
    &= -\E_t{\SurMinWith(t) + \min\{\surrogateoptitem[m], \SurMin(t)\}} + \inspcostitem[m] + \E_t{\SurMinWith(t + 1) + \min\{\priceitem[m], \SurMin(t + 1)\}} \\
    &\geq -\E_t{\min\{\surrogateoptitem[m], \SurMin(t)\}} + \inspcostitem[m] + \E_t{\min\{\priceitem[m], \SurMin(t)\}}.
  \]
  The inequality follows by~(O1), and its right-hand side is nonnegative by (O2) and \cref{lemma:min-r-surrogate}.
  \* If \iitem~$m$ is closed and the \dm\ selects it (without inspection), then
  \[
    \MoveEqLeft
    \E_t{K(t + 1)} - K(t) \\*
    &= -\E_t{\SurMinWith(t) + \min\{\surrogateoptitem[m], \SurMin(t)\}} + \E_t{\SurMinWith(t + 1) + \min\{\meanitem[m], \infty\}} \\
    &= -\E_t{\min\{\surrogateoptitem[m], \SurMin(t)\}} + \meanitem[m].
  \]
  The second equality follows by~(O1), and its right-hand side is nonnegative by (O2) and \cref{lemma:min-r-surrogate}. The infinity appears because $\SurMinWithout(t + 1) = \infty$ due to $m$ being selected at time~$t$.
  \* If \iitem~$m$ is open and the \dm\ selects it, then
  \[
    \MoveEqLeft
    \E_t{K(t + 1)} - K(t) \\*
    &= -\E_t{\SurMinWith(t) + \min\{\surrogateoptitem[m], \SurMin(t)\}} + \E_t{\SurMinWith(t + 1) + \min\{\priceitem[m], \infty\}} \\
    &= -\E_t{\min\{\surrogateoptitem[m], \SurMin(t)\}} + \priceitem[m].
  \]
  The second equality follows by~(O1), and its right-hand side is nonnegative because $\priceitem[m]$ is known to the \dm\ at time~$t$, so $\priceitem[m] = \E_t{\priceitem[m]}$. Again, the infinity appears because $\SurMinWithout(t + 1) = \infty$ due to $m$ being selected at time~$t$.
  \qedhere
  \*/
\end{proof}

\restate*{theorem:lh-combinatorial}
\begin{proof}
  We actually show the following stronger result: if for each item, a local hedging probability exists giving a local $\ratio$-approximation (or better), then using those same local hedging probabilities makes $\hedgepi$ a $\ratio \approxcombo$-approximation. This implies the theorem because by \cref{theorem:lh}, such local hedging probabilities exist for some $\alpha < 4/3$.

  Because the second stage of $\hedgepi$ runs $\policy$, the fact that $\E{\policycost} \leq \approxcombo \E{\realizedcostsurrogate}$ implies
  \[
    \E{\policycost[\hedgepi]} \leq \approxcombo \E{\realizedcostsurrogatehedge}.
  \]
  Combining this observation with \cref{theorem:combinatorial-cost-opt}, it suffices to show that for appropriately chosen local hedging probabilities, we have
  \[
    \label{eq:lh-combinatorial-goal}
    \E{\realizedcostsurrogatehedge} \leq \ratio \E{\realizedcostsurrogateopt}
  \]

  Before proving \cref{eq:lh-combinatorial-goal} formally, let us outline the main idea, which is essentially a generalization of the proof of \cref{theorem:lh-pandora}. Starting from $\realizedcostsurrogatehedge$ and replacing each item's \hedge-surrogate price with $\ratio$ times its \opt-surrogate price one by one, resulting in $\ratio \realizedcostsurrogateopt$. Because local hedging gives a local $\ratio$-approximation (\cref{definition:local-ratio-approx}) for each item, each replacement only increases the expected value.

  To formalize the above outline, we need notation for describing the replacement of surrogate prices one by one. To that end, let
  \[
    \surrogatebaseitem^{(m)}
    &= \begin{cases}
         \ratio \surrogateoptitem & \textnormal{if } \itemindex \leq m \\
         \surrogatehedgeitem & \textnormal{if } \itemindex > m,
       \end{cases}
    \\
    \realizedcostsurrogatebase^{(m)}
    &= \min_{\selectedset \in \constraints}{} \gp*{\sum_{\itemindex \in \selectedset} \surrogatebaseitem^{(m)} + \terminalcost(\selectedset)}.
  \]
  Then it suffices to show $\E{\realizedcostsurrogatebase^{(m - 1)}} \leq \E{\realizedcostsurrogatebase^{(m)}}$, as then
  \[
    \E{\realizedcostsurrogatehedge}
    = \E{\realizedcostsurrogatebase^{(0)}}
    \leq \dots
    \leq \E{\realizedcostsurrogatebase^{(\numitems)}}
    = \ratio \E{\realizedcostsurrogateopt},
  \]
  where the equalities at either end follow from \cref{definition:realized-cost-surrogate}.

  To show $\E{\realizedcostsurrogatebase^{(m - 1)}} \leq \E{\realizedcostsurrogatebase^{(m)}}$, the key idea is to split the minimization over $\itemsset \in \constraints$ into cases based on whether \iitem~$m$ is in the optimizing set. We therefore define
  \[
    \SurMinWithout &= \min_{\selectedset \in \constraints \,|\, m \not\in \selectedset}{} \gp*{
      \sum_{\itemindex \in \selectedset} \surrogatebaseitem^{(m)}
      + \terminalcost(\selectedset)
    },
    \\
    \SurMinWith &= \min_{\selectedset \in \constraints \,|\, m \in \selectedset}{} \gp[\Bigg]{
      \mkern0.75mu \smashoperator[r]{\sum_{\itemindex \in \selectedset \setminus \{m\}}}
      \surrogatebaseitem^{(m)} + \terminalcost(\selectedset)
    }.
  \]
  Because $\surrogatebaseitem^{(m - 1)} = \surrogatebaseitem^{(m)}$ for all $\itemindex \neq m$, we have
  \[
    \realizedcostsurrogatebase^{(m - 1)} &= \min\{\SurMinWithout, \SurMinWith + \surrogatehedgeitem[m]\}, \\
    \realizedcostsurrogatebase^{(m)} &= \min\{\SurMinWithout, \SurMinWith + \alpha \surrogateoptitem[m]\}.
  \]
  Note also that $\SurMinWithout$ and $\SurMinWith$ are independent of $\surrogatehedgeitem[m]$ and $\surrogateoptitem[m]$, as they depend on the surrogate prices of all \items\ \emph{except} \iitem~$m$. Using this and the fact that local hedging gives a local $\alpha$-approximation for \iitem~$m$ (\cref{definition:local-ratio-approx}), we compute
  \[
    \E{\realizedcostsurrogatebase^{(m - 1)} \given \SurMinWithout, \SurMinWith}
    &= \SurMinWith + \E{\min\{\surrogatehedgeitem[m], \SurMinWithout - \SurMinWith\} \given \SurMinWithout, \SurMinWith} \\
    &\leq \SurMinWith + \E{\min\{\alpha \surrogateoptitem[m], \SurMinWithout - \SurMinWith\} \given \SurMinWithout, \SurMinWith} \\
    &= \E{\realizedcostsurrogatebase^{(m)} \given \SurMinWithout, \SurMinWith},
  \]
  which implies $\E{\realizedcostsurrogatebase^{(m - 1)}} \leq \E{\realizedcostsurrogatebase^{(m)}}$, as desired.
\end{proof}

\section{Alternative proof of the lower bound for single-item selection}
\label{appendix:beyhaghi}

In this section, we give an alternative proof of \cref{theorem:cost-opt} by reducing it to a result of \citet{beyhaghi2019pandora}, which is a tighter lower bound on $\inf_\policy \E{\policycost}$. The reason we do not directly use \citeauthor{beyhaghi2019pandora}'s result for our proof is that, as will soon become clear, their result is less explicit than \cref{theorem:cost-opt}. Rather than a single expression that works for all policies~$\policy$, their result bounds $\E{\policycost}$ using an expression which itself depends on~$\policy$. But it is instructive to present their result and show how it implies our weaker but more explicit bound.

We note that although \citet{beyhaghi2019pandora} consider the reward-maximization setting instead of the cost-minimization setting, the definitions and proofs we consider in this section easily translate to the cost setting. We believe \citeauthor{beyhaghi2019pandora}'s approach could also be used to obtain a policy-dependent version of \cref{theorem:combinatorial-cost-opt}, but this is less immediate.

The main idea behind the approach \citet{beyhaghi2019pandora} take is this: to prove a policy-dependent bound, one should use a \emph{policy-dependent analogue of surrogate prices}. One way of doing so is the following.

\begin{definition}
  \label{definition:surrogate-pi}
  Let $\policy$ be a policy for the \optiona\ inspection problem. The \emph{$\policy$-surrogate price} of \iitem~$\itemindex$ is the random variable
  \[
    \surrogatebaseitem^\policy = \begin{cases}
      \surrogateitem & \textnormal{if $\policy$ inspects item~$\itemindex$} \\
      \mean & \textnormal{if $\policy$ does not inspect item~$\itemindex$}.
    \end{cases}
  \]
\end{definition}

Unlike the other types of surrogate prices we consider, two different \items' $\policy$-surrogate prices are \emph{not necessarily independent}. This is because the hidden price revealed by inspecting one \iitem\ may influence the decision of whether to inspect another.\footnote{%
  This is true even under local hedging, which uses independent randomness to commit to inspect-before-select or no-inspect for each item. The issue is that local hedging might not inspect an \iitem\ even if it is inspect-before-select, and this decision depends on the hidden prices of other \items.}
Despite this subtlety, \citet{beyhaghi2019pandora} use $\policy$-surrogate prices to prove the following bound.

\begin{proposition}[{cost analogue of \citet[Lemma~16]{beyhaghi2019pandora}}]
  \label{proposition:beyhaghi-bound}
  In the \optiona\ inspection problem, the expected total of policy~$\policy$ satisfies
  \[
    \E{\policycost} \geq \E*{\min_{\itemindex \in \itemsset} \surrogatebaseitem^\policy}.
  \]
\end{proposition}

As we show below, this result implies \cref{theorem:cost-opt}.

\begin{proof}[Alternative proof of \cref{theorem:cost-opt}]
  In light of \cref{proposition:beyhaghi-bound}, it suffices to show
  \[
    \label{eq:cost-opt-goal-via-beyhaghi-bound}
    \E*{\min_{\itemindex \in \itemsset} \surrogatebaseitem^\policy}
    \geq \E*{\min_{\itemindex \in \itemsset} \surrogateoptitem}.
  \]
  Similarly to the proof of \cref{theorem:cost-opt}, we show this by replacing $\policy$-surrogate prices with \opt-surrogate prices one-by-one. Specifically, it suffices to show that for each $\itemindex \in \itemsset$, we have
  \[
    \E{\min\curlgp{\dots, \surrogatebase_{\itemindex - 1}^\policy, \surrogatebaseitem^\policy, \surrogatebase_{\itemindex + 1}^\opt, \dots}}
    \geq \E{\min\curlgp{\dots, \surrogatebase_{\itemindex - 1}^\policy, \surrogatebaseitem^\opt, \surrogatebase_{\itemindex + 1}^\opt, \dots}},
  \]
  as chaining these inequalities together for all $n \in \itemsset$ yields \cref{eq:cost-opt-goal-via-beyhaghi-bound}. The above holds if for all $r \in \reals$,
  \[
    \E{\min\{\surrogatebaseitem^\policy, r\}}
    \geq \E{\min\{\surrogateoptitem, r\}}.
  \]
  By \cref{definition:surrogate-pi, lemma:min-r-surrogate}, we have, as desired,
  \[
    \E{\min\{\surrogatebaseitem^\policy, r\}}
    &\geq \min\curlgp[\big]{\E{\min\{\surrogateitem, r\}}, \min\{\meanitem, r\}} \\
    &= \min\curlgp[\big]{\E{\min\{\surrogateitem, r\}}, \meanitem} \\
    &= \E{\min\{\surrogateoptitem, r\}}.
    \qedhere
  \]
\end{proof}

\end{document}